\documentclass{colt2018}

\usepackage{times}
\usepackage{float}
\usepackage{algorithm}
\usepackage{makecell}
\usepackage{multirow}
\usepackage{hhline}
\usepackage[OT1]{fontenc} 
\usepackage{algorithmic}
\usepackage{mathrsfs}

\newtheorem{myTheo}{Theorem}

\usepackage{enumitem}
\usepackage{wrapfig}
\usepackage{tikz}

\begin{document}
\title{Spatial CUSUM for Signal Region Detection}
\author{\Name{Xin Zhang} \Email{xinzhang@iastate.edu}\\
 \addr Department of Statistics, Iowa State University
 \AND
 \Name{Zhengyuan Zhu} \Email{zhuz@iastate.edu}\\
 \addr Department of Statistics, Iowa State University 
}
\maketitle

\begin{abstract}
Detecting weak clustered signal in spatial data is important but challenging in applications such as medical image and epidemiology. 
A more efficient detection algorithm can provide more precise early warning, and effectively reduce the decision risk and cost. 
To date, many methods have been developed to detect signals with spatial structures.
However, most of the existing methods are either too conservative for weak signals or computationally too intensive.
In this paper, we consider a novel method named Spatial CUSUM (SCUSUM), which employs the idea of the CUSUM procedure and false discovery rate controlling. 
We develop theoretical properties of the method which indicates that asymptotically SCUSUM can reach high classification accuracy.
In the simulation study, we demonstrate that SCUSUM is sensitive to weak spatial signals.
This new method is applied to a real fMRI dataset as illustration, and more irregular weak spatial signals are detected in the images compared to some existing methods, including the conventional FDR, FDR$_L$ and scan statistics.
\end{abstract}


\section{Introduction} \label{sec:intro}

Spatial signal detection is an important topic in many fields, including astrophysics (\cite{abazajian2012detection}; \cite{gladders2000new}), brain imaging analysis (\cite{craddock2012whole}; \cite{zhang2011multiple}; \cite{blumensath2013spatially}; \cite{shen2013groupwise}), epidemiology (\cite{kulldorff1995spatial}; \cite{tango2000test}; \cite{wheeler2007comparison}), meteorology (\cite{sun2015false})etc.
Typically, given a spatial domain $\mathscr{D}$, e.g. an brain image or a geographical map, if there is no spatial signal, all the observations could be regarded to follow the same distribution.
While with the exising of spatial signals, the responses within a unknown sub-region are from a different distribution.
Locating signal regions with low signal-noise ratio is meaningful in the early detection and warning systems:
In the early stage of abnormality, the spatial signal is very weak compared with the measurement noise;
however, an accurate early warning could effectively reduce the decision risk and avoid unnecessary but lethal cost.
Such warning systems have been studied and applied in many practical cases, e.g. disease and weather monitoring (\cite{thomson2001development}; \cite{grover2005online}; \cite{breed2011weather}.)
Therefore, there will be a huge breakthrough if weak spatial signals can be efficiently identified.

Thus far, many methods and algorithms have been developed for spatial signal detection.
One class of approaches to identify spatial clusters is the spatial scan statistics (\cite{glaz2009scan}; \cite{glaz2001scan}; \cite{priebe2005scan}; \cite{glaz2012scan} etc.) 
Scan statistics, also known as window statistics, was first proposed in \cite{naus1965clustering}.
The idea is to perform likelihood ratio tests on all the scan windows of different sizes and locations and identify the significant windows as clusters.
This method was designed to find unusual clusters of randomly positioned points.
\cite{naus1982approximations} developed the asymptotic distribution for the scan statistics and proposed the method to find the maximum cluster of points on a line or circle, the length of the longest success run in Bernoulli trials, and the generalized birthday problem.
\cite{kulldorff1999spatial} extended the framework of the conventional scan statistics to multidimensional scenario, including two-dimensional scan statistics on the plane or on a sphere and three-dimensional scan statistics in space or in space—time.
However, if the the shape of the true cluster is not circle or ellipsoid, the power of the traditional scan statistics will significantly reduce.
Additionally, without $p$-value adjustment, the detection result from scan statistics might be too aggresive (\cite{zhang2010spatial}.)

Another branch of the detection methods is based on multiple testing and false discovery rate (FDR) controlling (\cite{benjamini1995controlling}; \cite{benjamini2001control}; \cite{genovese2002thresholding}; \cite{miller2001controlling}; \cite{zhang2011multiple}; \cite{tango2000test}; \cite{sun2015false}.)  Multiple hypothesis testing is concerned with testing several statistical hypotheses simultaneously, and false discovery rate is a criterion designed to control the expected proportion of rejected null hypotheses that are incorrect rejections:
\begin{equation}\label{FDR}
\text{FDR}=\mathbb{E}[\frac{\#\text{incorrect rejections}}{\#\text{rejected null hypotheses}}].
\end{equation}
In spatial signal detection, the statistical hypotheses are about whether locations belongs to signal region or not.
\cite{genovese2002thresholding} applied multiple tesing to functional 􏰔􏰃􏰓􏰂􏰅􏰈􏰕􏰎􏰊􏰈􏰉􏰊􏰔􏰃􏰓􏰂􏰅􏰈􏰕􏰎􏰊􏰈􏰉􏰊neuroimaging data and used FDR to find a threshold for signal classification. Their experiments showed that 􏰀􏰁􏰃FDR worked more conservatively when the correlations between hypotheses are high.
\cite{miller2001controlling} applied FDR procedure to astrophysical data, and showed that FDR had a similar rate of correct detections and signicantly less false detections compared with certain standard testing procedures.
In \cite{tango2000test}, multiple testing was used to detecting spatial disease clusters.
To improve detecting effectiveness, \cite{zhang2011multiple} proposed a testing procedure named FDR$_L$ with the consideration of the spatial structures. 
By aggregating the local $p$-values, FDR$_L$ could avoid the lack of identification phenomenon and improve the detection sensitivity.
\cite{sun2015false} developed an oracle procedure, which optimally controlled the false discovery rate, false discovery exceedance and false cluster rate, for multiple testing of spatial signals. 
The tropospheric ozone data in eastern USA were analyzed with their method to show the detection effectiveness.
Although FDR and its variants have good statistical interperation and can be easily implemented, the prior knowledge about the null distibution is required and many signals can be missed when signal-noise ratio is small.

In this paper, we will introduce a novel detecting method named Spatial CUSUM (SCUSUM) to identify spatial signal regions.
We assume the expected value of the signal regions is different (usually higher) from that of the indifference regions, and noise processes are zero-mean and independent.
SCUSUM has two steps: First applying moving window and CUSUM cut-off to estimate signal weight for each location; then determining a threshold with FDR controlling. 
Moving window method has been broadly used to analyze temporal and spatial data (\cite{paez2008moving}; \cite{haas1990kriging}), where the neighboring data is utilized to capture local feature.
Similarly, in our work, the spatial data is projected into an array, which can be analyzed by the CUSUM procedure.
The CUSUM procedure (or cumulative summation) is well-known to locate changepoints in time series (\cite{horvath2012change}; \cite{cho2016change}; \cite{gromenko2017detection}; \cite{wang2018high}; \cite{aue2009estimation}.)
With the CUSUM transfromation, the testing statistics can be compared with the standard Brownian Bridge to test the existence of the changepoint. 
The location of changepoint is where CUSUM reaches the maximum.
In our work, by repeating CUSUM cut-off with moving window, the detection frequency can be calculated for each location.
We define this frequency as signal weight.
Then, the null density and alternative density can be approximated by bounded density estimation method. 
A proper threshold could be found based on FDR to identify the singal region.
Our theoretical results show that SCUSUM could asymptotically reduce the misclassification rate to zero with probability $1.$
The experiment section support that our method could detect more weak spatial signals, compared with the existing methods, including FDR, FDR$_L$ and scan statistics.

The rest of the paper is organized as follows. 
The spatial signal detection problem is formulated mathematically in Section 2. 
The details of our proposed method are introduced in Section 3. 
Section 4 presents simulation comparisons between SCUSUM and FDR$_L$ under different signal strengths and noise dependence structures. 
An application of four methods (SCUSUM, scan statistics, FDR and FDR$_L$), to a real fMRI data is given in Section 5.
We give the conclusions in Section 6.
The proofs are shown in Appendix.

\section{Problem Formulation} \label{sec:formulation}

Let $\mathscr{D}$ be the entire spatial domain, $s$ present the location belonging to $\mathscr{D}$ and $ x(s)$ be the observed data at location $s.$
Consider $\mathscr{D}_{\mathscr{A}}$ be the signal region in $\mathscr{D}$, and its complement set $\mathscr{D}_{\mathscr{A}}^c$ be the indifference region.
We assume that under $H_0,$ there is no signal region (i.e. $\mathscr{D}_{\mathscr{A}}=\varnothing$) and $x(s)$ has the same mean process $\mu$; while under $H_1,$ $x(s)$ has mean $\mu_1$ if $s\in \mathscr{D}_{\mathscr{A}}$ and $\mu_0$ if $s\in \mathscr{D}_{\mathscr{A}}^c.$

Hence, the following additive model for the random variables $X=\{x(s),s\in \mathscr{D}\}$is considered:
\begin{equation}\label{basic_model}
x(s)=\mu_0\mathbb{I}(s\in \mathscr{D}_{\mathscr{A}}^c)+\mu_1\mathbb{I}(s\in \mathscr{D}_{\mathscr{A}})+\epsilon(s),
\end{equation}
where both $\mu_0$ and $\mu_1$ are the unobserved mean (w.l.o.g, assume $\mu_1\ge\mu_0$,) $\epsilon(s)$ is the independent noise and $\mathbb{I}(\cdot)$ is the indicator function. 
Note here, we don't assign any distribution model to the noise. 
The only requirement for noise is that it has zero mean and i.i.d.
Our goal is to identify the signal region $\mathscr{D}_{\mathscr{A}}.$

\section{Proposed method} \label{sec:method}

In this section, we will give the details of our proposed method named Spatial CUSUM (SCUSUM), which has two steps: 1) For each location, we first estimate the signal weight, which is expected to be large in the signal region $\mathscr{D}_{\mathscr{A}}$, while small in the indifference region $\mathscr{D}_{\mathscr{A}}^c$. 
2) Given a significant level $\alpha,$ a threshold is determined based on FDR idea. 
In section \ref{step1}, we describe the way of using the moving window idea to project a spatial domain into a sequence and then estimating the signal weight for each location with the CUSUM cut-off. 
In section \ref{step2}, we introduce how to estimate the null distribution $f_{H_0}$ and alternative distribution $f_{H_1}$ with estimated signal weights, following by the step to determine the detection threshold.
In section \ref{Neightsize}, we briefly discuss neighbor size selection for moving window in the first step.

\subsection{The first step: signal weight estimation}\label{step1}

Our signal weight estimation method is inspired by the CUSUM procedure for changepoint detection in time series. 
However, for spatial signal detection, the conventional CUSUM is impractical, mainly due to lack of natural order for spatial observations, which are located in $\mathbb{R}^2$ or $\mathbb{R}^3$ (in our work, we focus on $\mathbb{R}^2$.)
Hence, we consider to use the moving window technique to construct proper sequences.

\begin{figure}[t!]
\centering
\begin{tabular}{@{}ccc@{}}
\includegraphics[width=4cm]{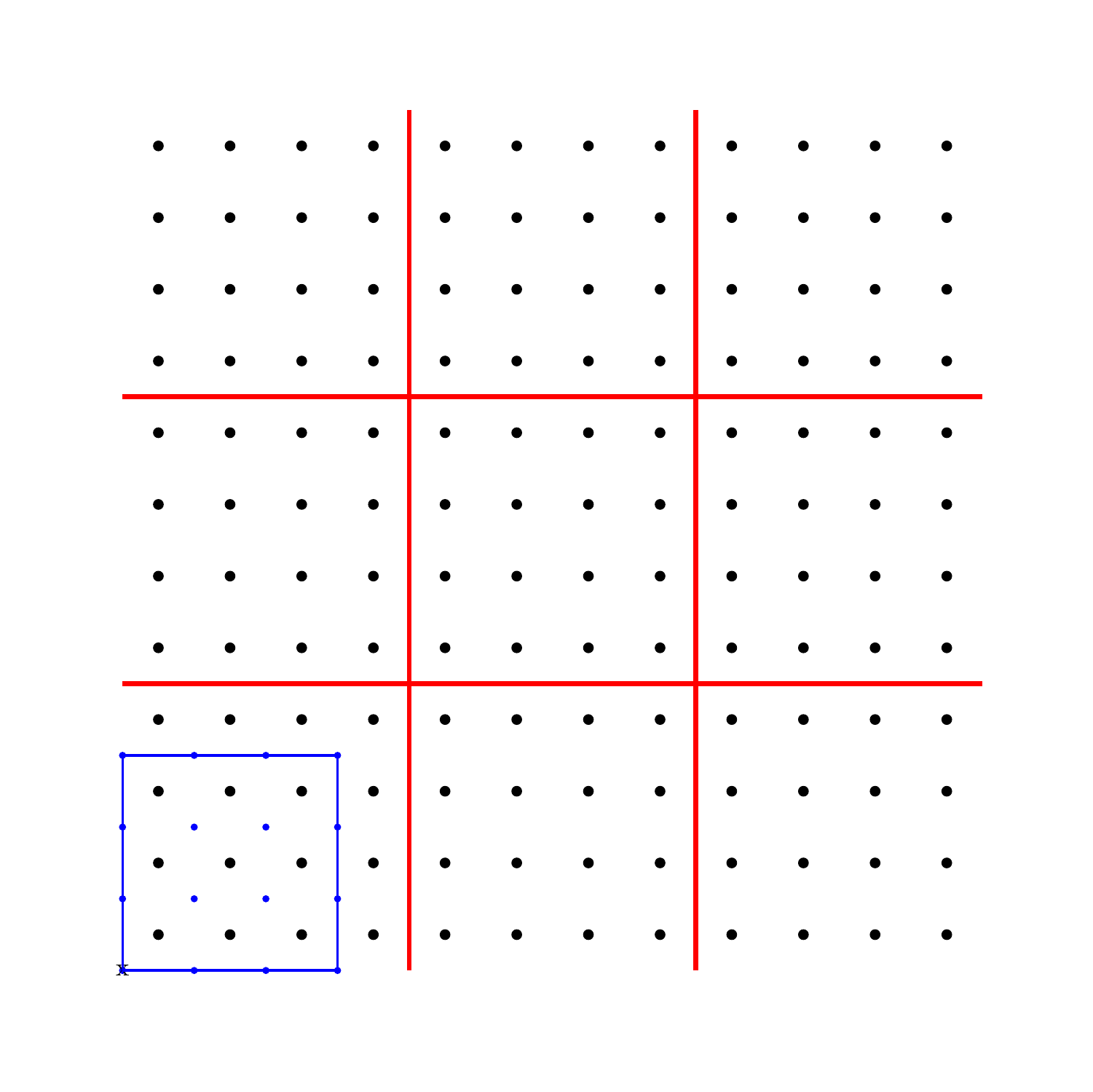}&
\includegraphics[width=4cm]{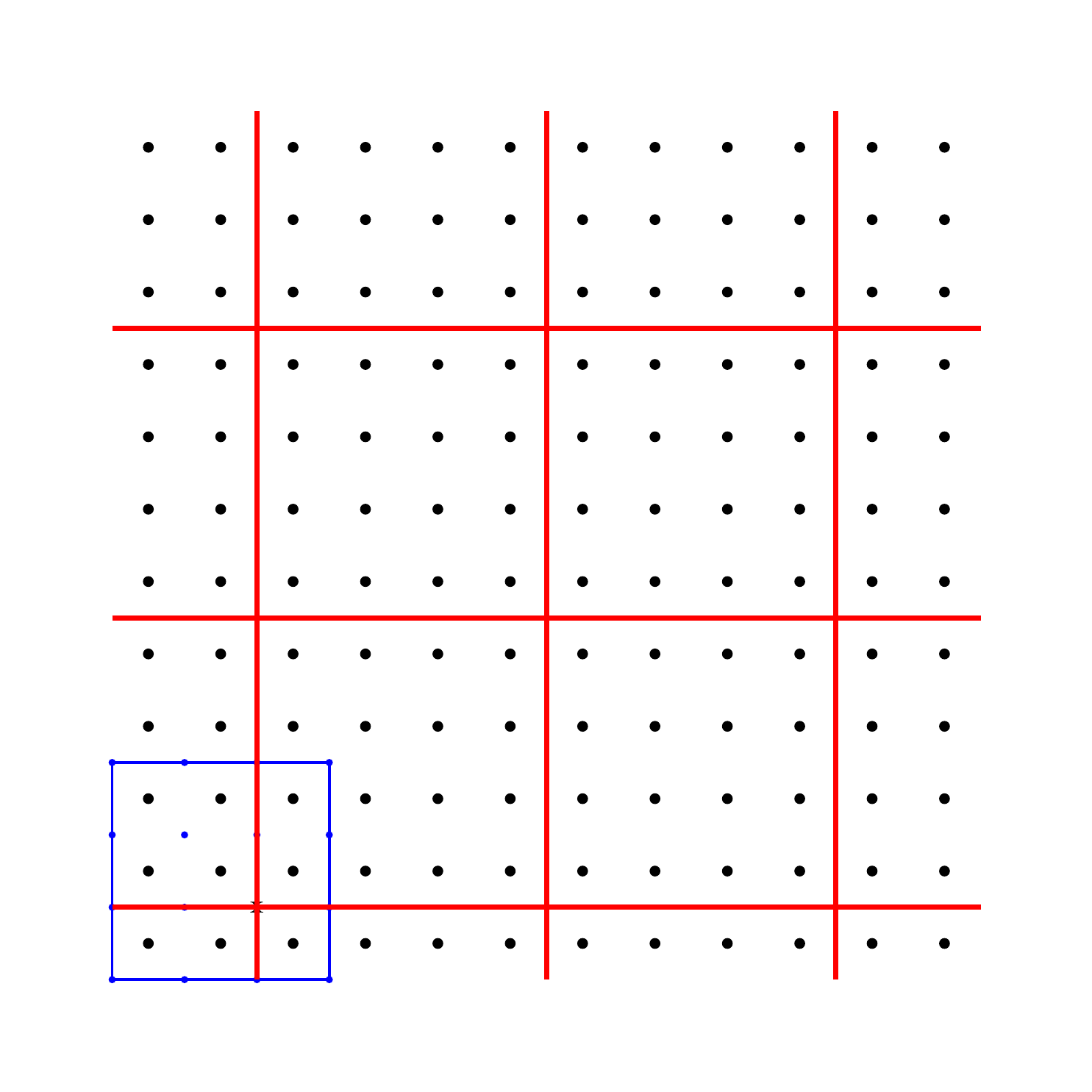}& 
\includegraphics[width=4cm]{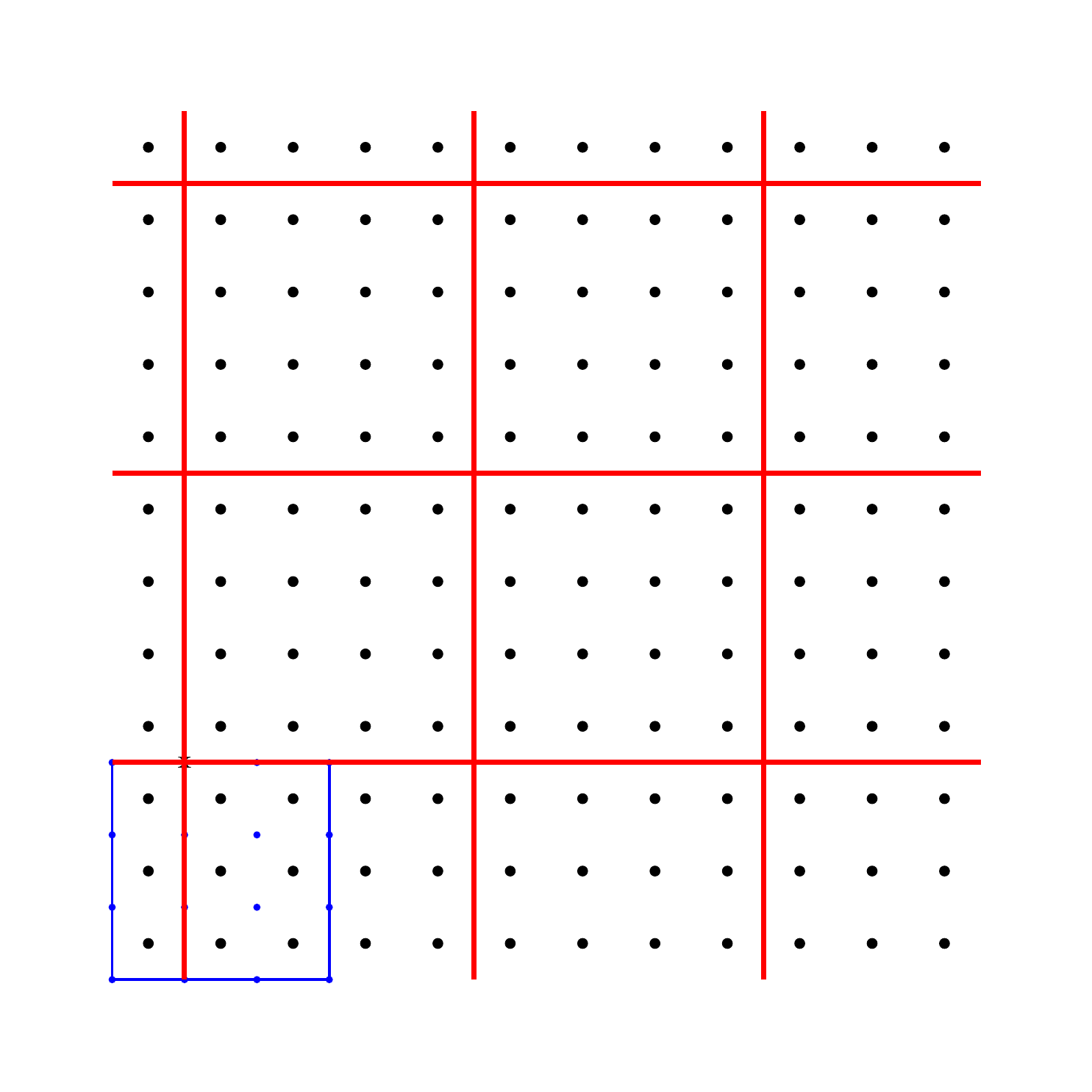}\\ 
\end{tabular}
\caption{Moving window to divide spatial domain: The blue points inside blue square is $\{(p+\frac{1}{2},q+\frac{1}{2}),~p=1:k,~q=1:k\}$, we select initial grid point from the set; The black points present the observed locations; red grids are boundary lines for blocks.} \label{fig:MW_grid}
\end{figure}

\begin{wrapfigure}{R}{0.4\columnwidth}
\centering{}
\includegraphics[width=1\linewidth]{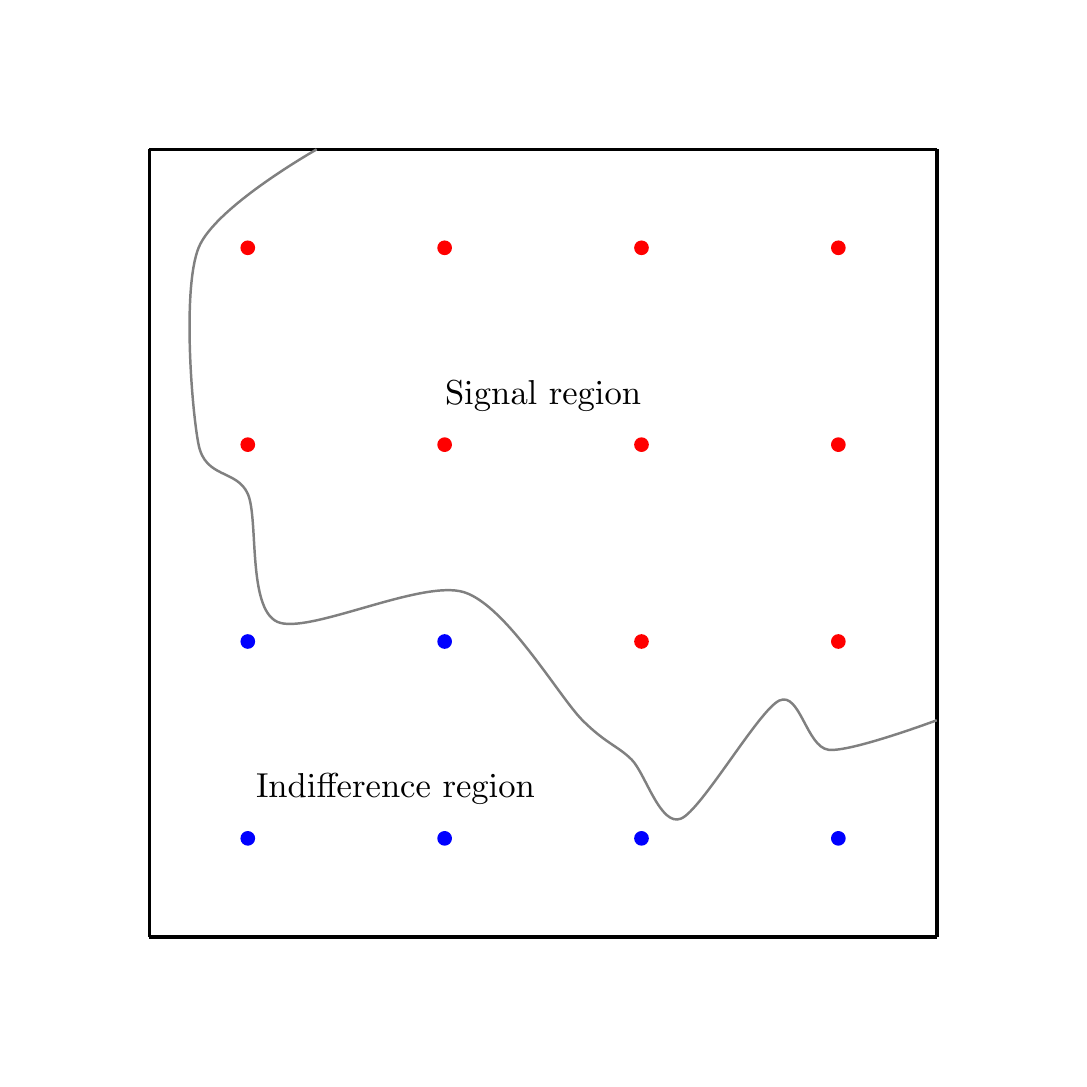}
\vspace{-0.2in}
\caption{The block at the boundary of signal region and indifference region. In this case, $p_i=0.625.$}
\label{fig:boundary}
\vspace{-.3in}
\end{wrapfigure}

Given a square spatial domain $\mathscr{D}$ and a neighbor size $k,$ we can select one point $(x,y)$ from $\{(p+\frac{1}{2},q+\frac{1}{2}),~p=1:k,~q=1:k\}.$ 
Then we can divide $\mathscr{D}$ into $b$ non-overlapping blocks $\{B_i\}_{i=1}^{b}$ of size $k\times k,$ with $(x,y)$ as the initial grid point (see Figure \ref{fig:MW_grid}.) 
So we have $n=\sum_{i=1}^{b} n_i,$ where $n$ is the total number of observations and $n_i$ is the observations in the block $B_i.$

According to model \ref{basic_model}, if block $B_i$ is inside the signal region, i.e. $B_i\in \mathscr{D}_{\mathscr{A}},$ then all the observations in it follow $x(s)=\mu_1+\epsilon(s), ~ \forall x\in B_i;$ 
if $B_i$ is inside the indifference region, i.e. $B_i\in \mathscr{D}_{\mathscr{A}}^c,$ then $x(s)=\mu_0+\epsilon(s), ~ \forall x\in B_i;$ 
if $B_i$ is at the boundary of $\mathscr{D}_{\mathscr{A}}$ and $\mathscr{D}_{\mathscr{A}}^c,$ then the observations $x(s)$ follows a mixture model: $x(s)=\mu_1+\epsilon(s)$ with probability $p_i$ and $x(s)=\mu_0+\epsilon(s)$ with $1-p_i$, where $p_i$ is the ratio of signal points inside $B_i$ (see Figure \ref{fig:boundary}.) 

Next, based on the above division, we construct two sequences to capture the feature of these spatial observations. 
The first sequence is sample sequence: A random sample, denoted as $\gamma_i$, is drawn from block $B_i$. It could be regarded as the 'representative' for this block. 
Meanwhile, we construct the second sequence by computing the block mean without the 'representative', $\tilde{\mu}_i=\frac{\sum_{x\in B_i}x-\gamma_i}{\sum_{x\in \mathscr{D}}\mathbb{I}(x\in B_i)-1}.$ 
As the number of observations $n_i$ in $B_i$ increases, the pseudo block mean gets closer to the true block mean, i.e. $\tilde{\mu}_i\rightarrow \frac{\sum_{x\in B_i}x}{\sum_{x\in \mathscr{D}}\mathbb{I}(x\in B_i)}.$ 
Hence, the pseudo block mean $\tilde{\mu}_i$ could present the local block mean.
Based on the analysis in the last paragraph, we could easily derive the following results for $\{\gamma_i\}_{i=1}^{b}$ and $\{\tilde{\mu}_i\}_{i=1}^{b}:$
\begin{align}
\gamma_i =&\left\{
\begin{aligned}
&\mu_1+\epsilon,~ \text{if }B_i \in \mathscr{D}_{\mathscr{A}}  \\
&\mu_0+\epsilon,~ \text{if }B_i \in \mathscr{D}_{\mathscr{A}}^c \\
&\mu_1z+\mu_0(1-z)+\epsilon, ~\text{if } B_i \text{ at boundary}
\end{aligned}\label{gammai}
\right.\\
&\notag \\
\mathbb{E}[\tilde{\mu}_i] &\left\{
\begin{aligned}
&=\mu_1,~ \text{if }B_i \in \mathscr{D}_{\mathscr{A}}  \\
&=\mu_0,~ \text{if }B_i \in \mathscr{D}_{\mathscr{A}}^c \\
&\approx p_i\mu_1+(1-p_i)\mu_0, ~\text{if } B_i \text{ at boundary}
\end{aligned}\label{tilde_mu}
\right.
\end{align}
where $z\sim Ber(p_i)$. 
(\ref{gammai}) and (\ref{tilde_mu}) show that even though $\{\gamma_i\}_{i=1}^{b}$ and $\{\tilde{\mu}_i\}_{i=1}^{b}$ are indepedent (see Lemma \ref{lemma.1}), they have similar patterns:
the closer $B_i$ is to $\mathscr{D}_{\mathscr{A}},$ the more likely it has large $\gamma_i$ and $\tilde{\mu}_i,$ and vice versa. 
Hence, we could consider to rearrange $\{\gamma_i\}_{i=1}^{b}$ according to the decreasing order of $\{\tilde{\mu}_i\}_{i=1}^{b}$, denoted as $\{\gamma^*_i\}_{i=1}^{b}.$
Intuitively, if there is no signal region, then $\{\gamma^*_i\}$ should be around $\mu_0;$ otherwise, $\{\gamma^*_i\}_{i=1}^{b}$ should have three parts: 
the first part presenting signal blocks is around $\mu_1$, the second part is the interim from $\mu_1$ to $\mu_0$ and the third part is indifference blocks around $\mu_0$ (see Figure \ref{fig:Ill_gamma}.) 

\begin{lemma}\label{lemma.1}
Based on model \ref{basic_model}, $\{\gamma_i\}$ and $\{\tilde{\mu}_i\}$ are indepedent. As the number of observations in each block $n_i$ goes to infinity, i.e. $\min n_i\rightarrow \infty,$ we have the following: under the null hypothesis $H_0:$ there is no signal region, then $\{\gamma^*_i\}$ is an i.i.d sequence; under the alternative hypothesis $H_1:$ signal region exists, then there exists $l_1$ and $l_2$ with $0 \le l_1 <l_2 \le b,$ 
\begin{align}
\mathbb{E}[\gamma^*_i]  = \left\{
\begin{aligned}
&\mu_1,~ \text{if } 0 \le i < l_1  \\
&\in (\mu_0, \mu_1) ~\text{if } l_1\le i < l_2\\
&\mu_0,~ \text{if } l_2 \le i \le n,
\end{aligned}\label{gamma*}
\right.
\end{align}
and $l_1$ is the number of blocks inside $\mathscr{D}_{\mathscr{A}},$ and $(b-l_2)$ is the number of blocks inside $\mathscr{D}_{\mathscr{A}}.$
\end{lemma}  

\begin{figure}[t!]
\centering
\begin{tabular}{@{}ccc@{}}
\includegraphics[width=6cm]{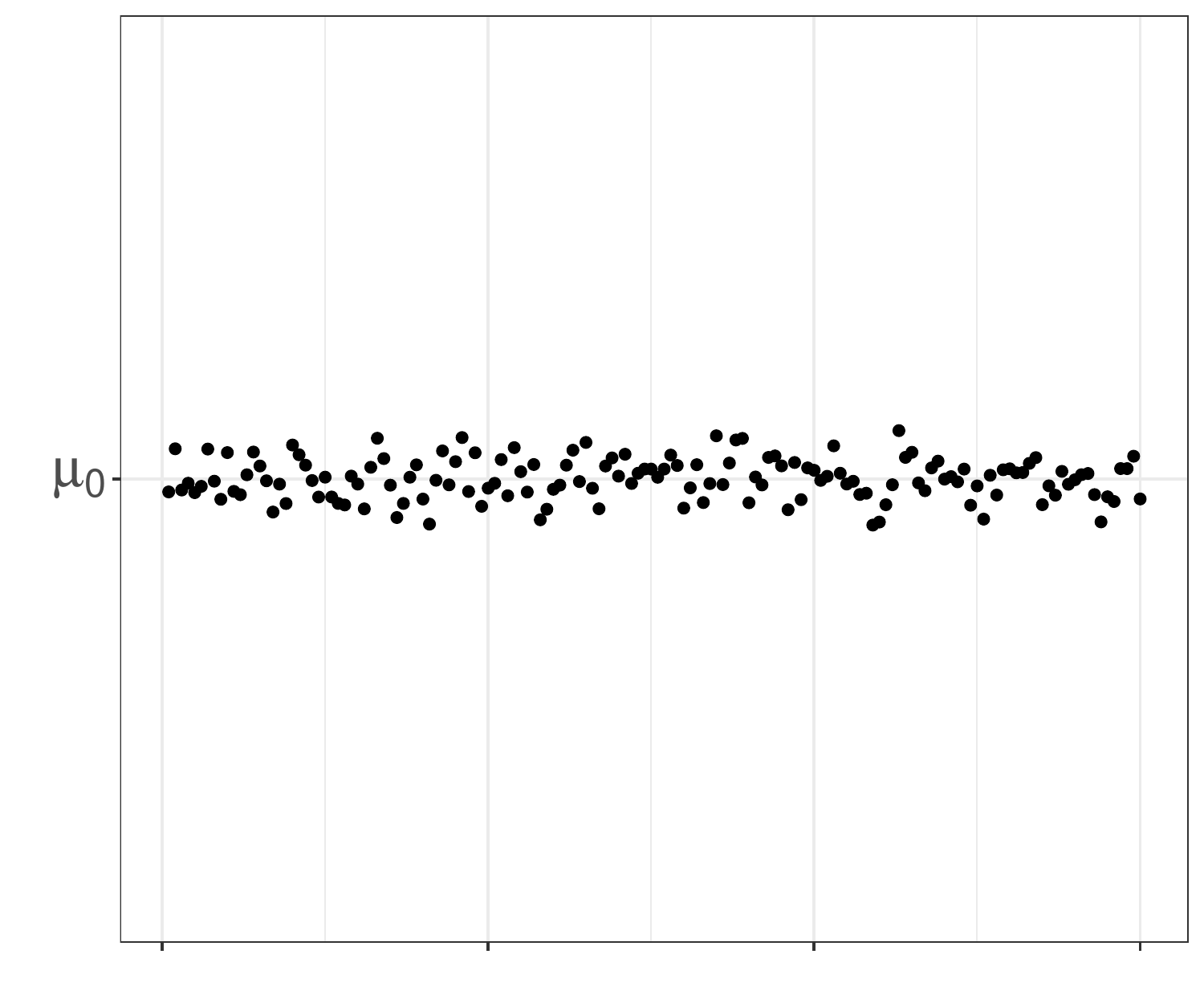}&
& 
\includegraphics[width=6cm]{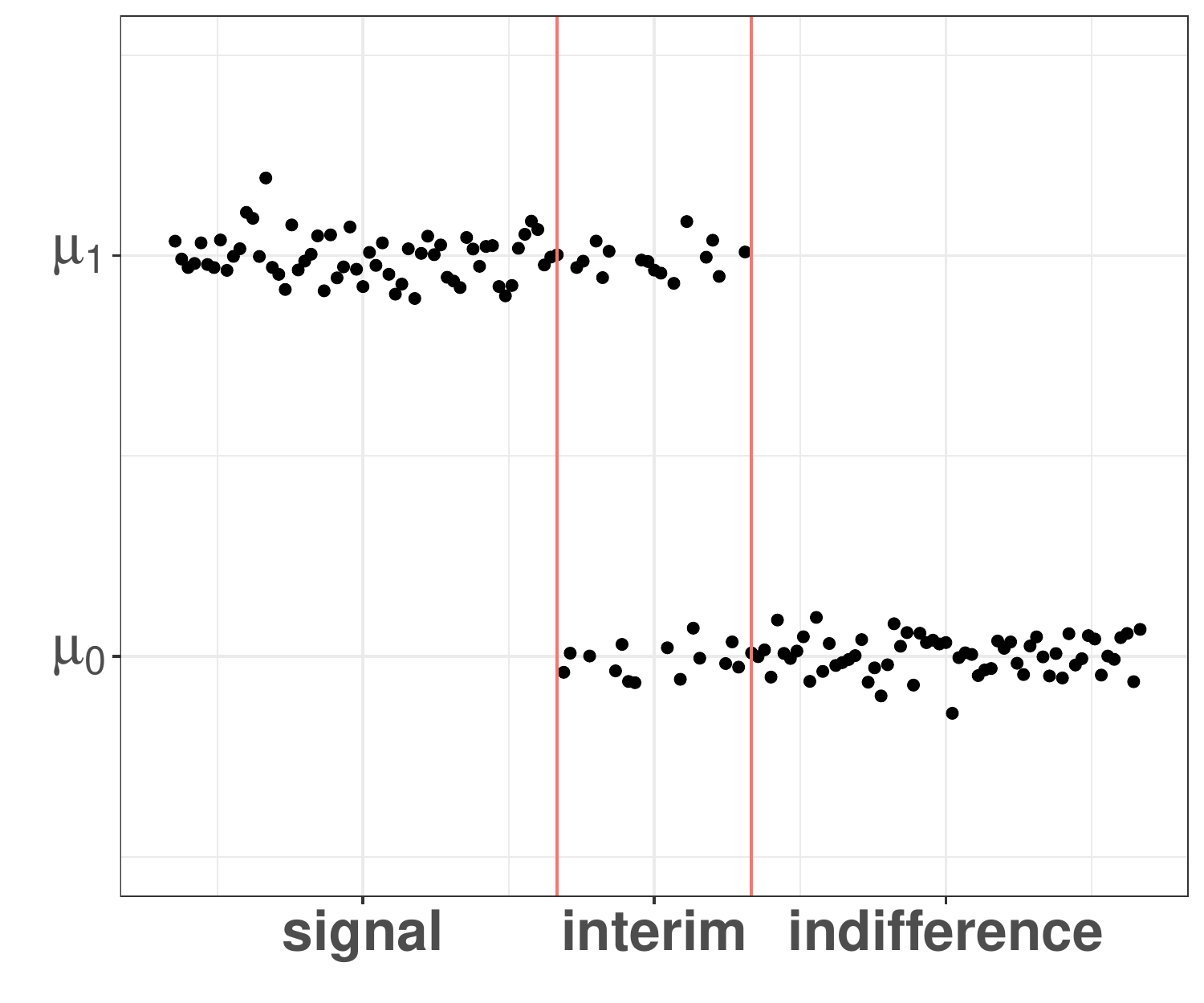}\\ 
(a) under $H_0$ &
&
(b) under $H_1$
\end{tabular}
\caption{The possible patterns of $\{\gamma_i^*\}_{i=1}^{b}:$ (a) presents the scenario without signal region and $\{\gamma_{(i)}^*\}_{i=1}^{b}$ are around $\mu_0;$
(b) shows the pattern with signal region and there are three parts: signal, interim and indefference.} \label{fig:Ill_gamma}
\end{figure}

Lemma \ref{lemma.1} shows that under $H_1,$ the projected sequence $\{\gamma^*_i\}_{i=1}^{b}$ has a changepoint in $[l_1,l_2]$ and the conventional CUSUM could help locate a cut-off index near or inside the interval. First, we compute the CUSUM statistics for $\{\gamma^*_i\}_{i=1}^{b}$ at each location:
\begin{equation}
\tilde{\gamma}_r=|\sum_{i=1}^{r} \gamma^*_i-\frac{r}{b}\sum_{i=1}^{b}\gamma^*_i|.
\end{equation}
Then the cut-off index is $t=\arg\max_{i} \tilde{\gamma}_i.$ The following theorem guarantees the accuracy of the cut-off index.
\begin{myTheo}\label{Theo.1}
Under the alternative hypothesis: if the signal region exists, as the number of block $b$ and the number of observations in each block $n_i$ go to infinity, then the cut-off index $t$ based on CUSUM procedure will fall into the interval $[l_1,l_2]$ with probability 1, i.e. $\mathbb{P}(l_1\le t \le l_2)=1$ as $b \rightarrow \infty$ and $\min n_i \rightarrow \infty$. 
\end{myTheo}
Theorem \ref{Theo.1} ensures that the cut-off procedure could asymptotically separate the signal region and indifference (see Figure \ref{fig:Ill_gamma} (b)).
Also, with the given spatial domain, as the block size becomes finer (equivalent to $b\rightarrow \infty$), the number of the blocks at the boundary is decreasing. 
Hence we have $(l_2-l_1)/b\rightarrow 0.$
Combining the results from Theorem \ref{Theo.1}, the number of misclassified locations goes to zero.

Above theoretical results require $b \rightarrow \infty$ and $\min n_i \rightarrow \infty$. 
In practice, with limited observations, the detected result might be affected by the initial point selection, especially when the signal region is not regular. 
Thus, we could eliminate the effect of initial point by going through all the possible inital points (see Figure \ref{fig:MW_grid}.)
We summarize our method in Algorithm \ref{Alg:MWDM}.
Also, we could sufficiently extract the local information and eliminate the effect of randomly sampled 'representatives' by repeated Algorithm \ref{Alg:MWDM} more than once. 
With these steps, we could estimate the signal weights $\{w(s)\}$ (or $\{\tilde{w}(s)\}$) by computing the detected frequency of each location. The larger signal weight means the location is more likely to belong to the signal region.

\begin{algorithm}[t!]
\caption{Moving Window detecting method for signal weight}\label{Alg:MWDM}
\begin{algorithmic} [1]
\REQUIRE observed data $\{x(s)\}$ on grid $\{(p,q),~p=1:n,~q=1:n\}$, neighbor size $k$, repeat times $m$;
\ENSURE corresponding signal weights $\{w(s)\}$ or $\{\tilde{w}(s)\};$
\FOR{$(x,y)$ in $\{(p+\frac{1}{2},q+\frac{1}{2}),~p=1:k,~q=1:k\}$}
    \STATE Divide $\mathscr{D}$ into blocks $\{B_i\}_{i=1}^{b}$ of size $k\times k$ based on $(x,y);$
    \STATE Sample one observation from each block $\gamma_i;$
    \STATE Estimate the block mean $\tilde{\mu}_i;$
    \STATE Reorder $\{\gamma_i\}_{i=1}^{b}$ according to $\{\tilde{\mu}_i\}_{i=1}^{b}$ decreasingly as $\{\gamma_i^*\}_{i=1}^{b};$
    \STATE Conduct CUSUM transformation on $\{\gamma_{(i)}^*\}_{i=1}^{b}$ as $\{\tilde{\gamma}_{(i)}^*\}_{i=1}^{b};$
    \STATE Find the location $t$ where $\{\tilde{\gamma}_i^*\}_{i=1}^{b}$ reaches maximum;
    \STATE Define the blocks corresponding to the first $t$ elements in $\{\gamma_i\}_{i=1}^{b}$ as signal block, and so do the observations in these blocks;
\ENDFOR
\STATE Compute corresponding signal weight $w(s)=\frac{\text{detected times for }x(s)}{k^2};$\\
(Option)
\STATE Repeat above produce $m$ times and obtain $\{w^{i}(s)\}_{i=1}^{m};$\\
\STATE Compute the average signal weights at each location $\{\tilde{w}^{i}(s)\}:$ $\tilde{w}^{i}(s)=\sum_{i=1}^{m}w^{i}(s)/m;$\\
\end{algorithmic}
\end{algorithm}

\subsection{The second step: Threshold estimation with FDR}\label{step2}

With the estimated signal weight and a given significant level $\alpha$, we could identify the signal region with FDR. 
Firstly, the weights are in $[0,1],$ and could be considered as the possibilities that the locations have signals.
Thus we could use density estimation with the boundary correction method to estimate the distribution of the signal weights, $f(x)$. Many density estimations have been studied in previous works (\cite{chen1999beta}; \cite{cowling1996pseudodata}; \cite{jones1996simple}; \cite{cattaneo2017lpdensity}.) In our work, we use the local polynomial density estimation method from \cite{cattaneo2017lpdensity}.

In the following, we analyze the characteristic of $f(x),$ which could help us estimate the threshold. Under the null hypothesis $H_0,$ with Lemma \ref{lemma.1}, we know that $\{\gamma_{(i)}^*\}_{i=1}^{b}$ are i.i.d. and the correponding blocks are random indexed. 
Hence with CUSUM cut-off procedure, the distribution for signal weight $f(x)$ is symmetric and has lower value with $x=0$ and $x=1$ (see Figure \ref{Fig:H0H1_fdr} (a).) 

\begin{lemma}\label{lemma.sym}
Under the null hypothesis $H_0,$ as the number of block $b$ and the number of observations in each block $n_i$ go to infinity, then the density for signal weights $f(x)$ is symmetric.
\end{lemma}

\begin{figure}[t!]
\centering
\begin{tabular}{@{}c@{}}
\includegraphics[width=15cm]{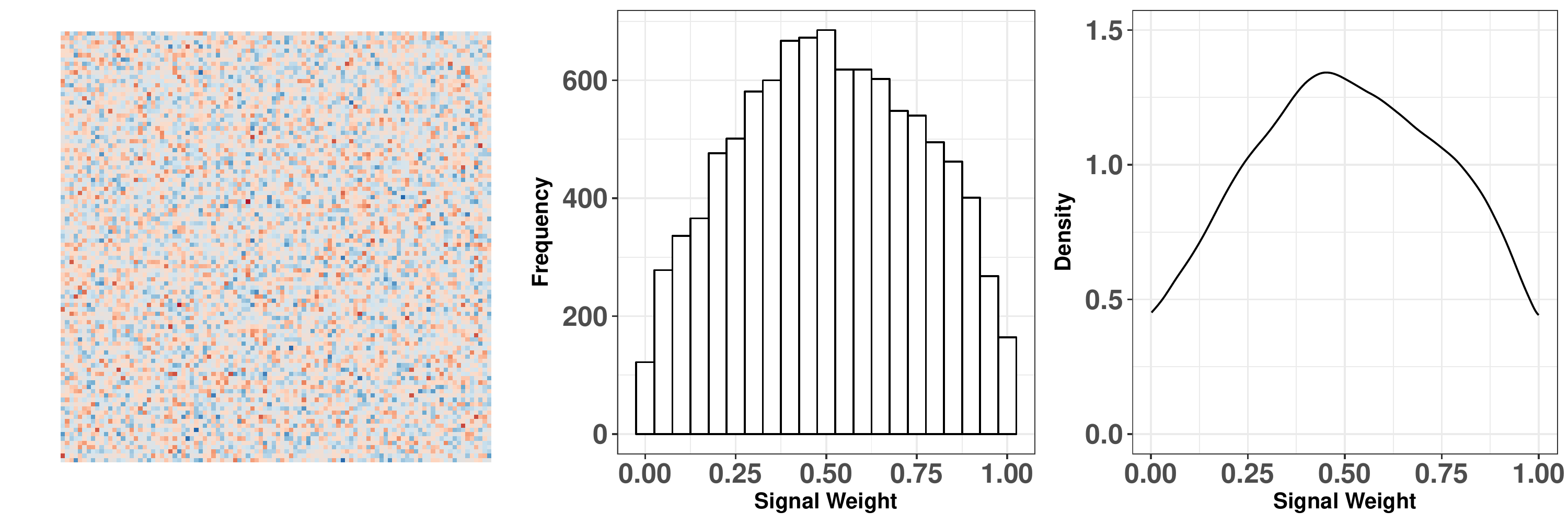}\\
(a) Under $H_0$\\
\includegraphics[width=15cm]{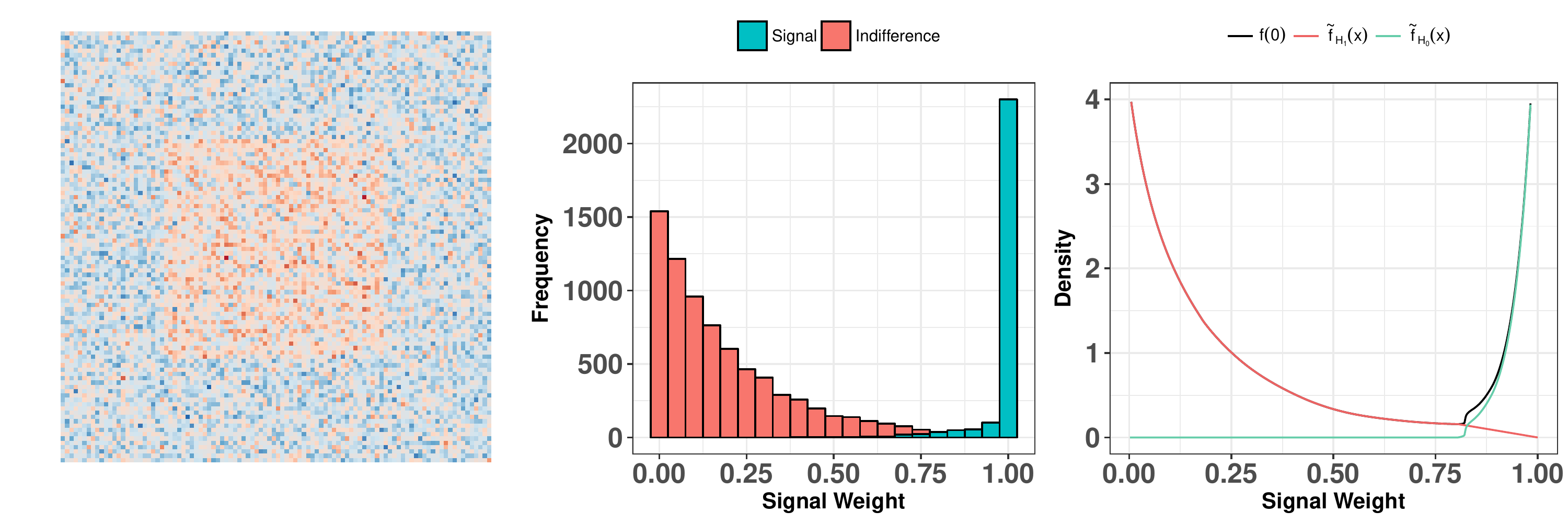}\\
(b) Under $H_1$ (signal-noise ratio is $1$)\\
\end{tabular}
\caption{The idea for the second step: (a) Under $H_0,$ there is no signal (first column.) The histogram of signal weights are shown in the second column and it's symmetric with 'peak' around $0.5.$ The estimated density is shown in the third column; (b) Under $H_1$, the histogram (the second column) is composited with two parts and has higher values around boundaries $0$ and $1,$ lower values at the middle part. The estimated densities are shownd in third column: the black curve is $f(x),$ green one is estimated null density $\tilde{f}_{H_0}(x)$ and red one is estimated alternative density $\tilde{f}_{H_1}(x)=f(x)-\tilde{f}_{H_0}(x).$}\label{Fig:H0H1_fdr}
\end{figure}

Under the alternative hypothesis $H_1,$ the distribution $f(x)$ should be composited by the null part $f_{H_0}(x)$ and alternative part $f_{H_1}(x):$ $f(x)=f_{H_0}(x)+f_{H_1}(x).$
The observations inside the signal region are more likely to be detected, i.e. the corresponding signal weight gets close to 1 and vice versa for the observations inside the indifference region. 
With finer block division, the fraction of the observations in the blocks at the boundary goes to 0. 
Therefore, $f_{H_1}(x)$ has the 'peak' near $1$ and $f_{H_0}(x)$ has the 'peak' around $0,$ which implies that $f(x)$ has two 'peaks' near the boundaries seperately and a 'valley' in the middle of $[0,1].$ Then we can use the line search to locate the 'valley', say $(t^*,f(t^*)),$ and conduct linear interpolation between the two points $(t^*,f(t^*))$ and $(1,0).$  
Obviously, the null density $f_{H_0}(x)$ is controlled by 
\begin{align}
\tilde{f}_{H_0}(x) =&\left\{
\begin{aligned}
&f(x), \text{ if } 0\le x\le t^* \\
&f(t^*)(1-\frac{x-t^*}{1-t^*}), \text{ if }  t^* < x \le 1
\end{aligned}\label{tilde_f_h0}
\right.
\end{align}
which could be used as estimated null density (see Figure \ref{Fig:H0H1_fdr} (b).)
Recall the definition of the marginal false discovery rate (mFDR) (\cite{genovese2002operating}; \cite{sun2015false}):
\begin{equation}\label{Eq:mFDR}
\text{mFDR}=\frac{\mathbb{E}[\#\text{false positive}]}{\mathbb{E}[\#\text{rejected}]}.
\end{equation}
Thus, we could control mFDR with given significant level $\alpha$ by finding a threshold $c$ so that 
\begin{equation}\label{Eq:thre}
c=\arg\min_x (\frac{\tilde{f}_{H_0}(x)}{f(x)}\le \alpha).
\end{equation}
The observations with signal weight larger than $c$ are the detected signals.
This step is summarized in Algorithm \ref{Alg:SRDM}.
\begin{algorithm}[t!]
\caption{The signal region detection method}\label{Alg:SRDM}
\begin{algorithmic} [1]
\REQUIRE Signal weights $\{w(s)\}$ or $\{\tilde{w}^{i}(s)\}$, signficant level $\alpha$;
\ENSURE corresponding detected result;
\STATE {}Estimate the density curve $f(x)$ based on signal weights, $x \in [0,1];$
\STATE Estimate null density $f_{H_0}(x)$ and alternative density $f_{H_1}(x)$ with (\ref{tilde_f_h0});
\STATE Compute mFDR and find the threshold $c$ with (\ref{Eq:thre});
\STATE Obtain the detected result with the threshold $c$;
\end{algorithmic}
\end{algorithm}

\subsection{Neighbor size selection}\label{Neightsize}

In this part, we consider the selection of neighbor size $k,$ and this mainly affects the accuracy in Section \ref{step1}. Intuitively, the larger $k$ means the larger block and tends to over-smooth; while the smaller $k$ implying the small block might lose spatial information.

In order to make the "right" cut-off, we need to ensure two points: 
1) the variance of the pseudo block mean $\{\tilde{\mu}_i\}$ should be as small as possible, so that we could reasonably rearrange the 'representatives' $\{\gamma_i\};$ 
2) the length of the rearranged 'representative' sequence $\{\gamma^*_i\}$ should be as long as possible, which could ensure the cut-off location $t$ fall into $[l_1,l_2]$ with probability $1$. 
For the first point, we need to make the number of observations in each block $n_i$ go to infinity; for the second point, the length of $\{\gamma^*_i\}$ is the number of blocks $b.$ 
And the relationship between $b,$ $k$ and $n_i$ could be approximated as: 
\begin{align}
\left\{
\begin{aligned}
&\min n_i \approx k^2\\
&b \approx n/k^2
\end{aligned}\label{Eq:bkn}
\right.
\end{align}
Hence, we could get the following trade-off problem:
\begin{equation}
k_{opt} = \arg\min_k k^2+C_1\frac{n}{k^2}=\sqrt[4]{C_1n},
\end{equation}
where $C_1$ is a given weight to reflect which part we want to emphasize and $n$ is the total number of observations.

Of course, the above analysis is based on the theoretical result. In practice, the neighbor size selection depends on the specific problem and application. Related discussions on neighbor size selection could be found in existing works (\cite{wang2006neighborhood}; \cite{hall1995blocking}; \cite{sun2015false}.)

\section{Simulation Study} \label{sec:SimuStudy}
\begin{wrapfigure}{R}{0.35\columnwidth}
\centering{}
\includegraphics[width=1\linewidth]{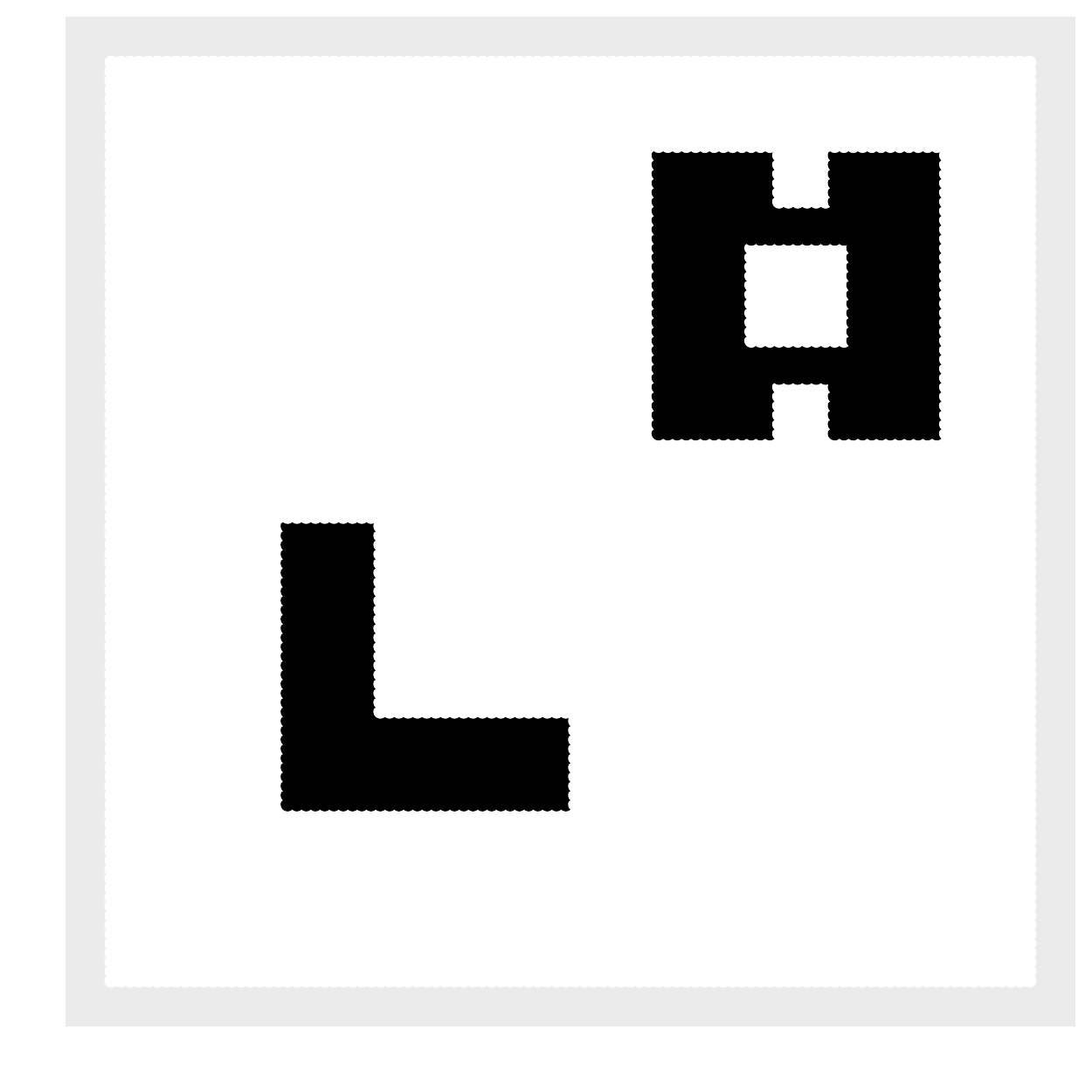}
\vspace{-0.2in}
\caption{Ground truth of our simulation setting}
\label{fig:groundtruth}
\vspace{-.3in}
\end{wrapfigure}

In this section, we will use simulation to show the effectiveness of our proposed method.
We compare SCUSUM with FDR$_L,$ because the two methods are designed to detect irregular signals with false discovery rate controlling.
All the examples are simulated in the image with $100\times 100$ pixels. 
Although in model \ref{basic_model} we didn't specify the distribution for noise process, we consider independent standard normal distribution $N(0,1)$ for noise term and generate the data according to the model:
\begin{equation}\label{simu_model}
x(i,j)=\mu(i,j)+\epsilon(i,j),~i,j=1,...,100,
\end{equation}
where $\mu(i,j)=0$ for $(i,j)\in\mathscr{D}_{\mathscr{A}}^c,$ and $\mu(i,j)\neq 0$ for $(i,j)\in\mathscr{D}_{\mathscr{A}},$ the 'L' shape and 'H' shape shown in Figure \ref{fig:groundtruth}: the black region is signal region $\mathscr{D}_{\mathscr{A}}$ and white region is the indifference part..
The total number of signal pixels is $1288.$ 
Here we mainly concern about the accuracy of classification, both false positive and false negative.
We set the repeated time $m$ in Algorithm \ref{Alg:MWDM} as $50.$ 
Additional, for FDR$_L$, we do a standard normal test on each pixels and then apply the algorithm on the corresponding $p$-values.

In Table \ref{Table:Compare}, the simulation results for SCUSUM and FDR$_L$ are shown. Under each setting, we repeat simulation $100$ times. 
For the two algorithms, we preset the significant level $\alpha$ as $0.05.$ In the siganl region, $\mu$ ranges from $0.8$ to $2.$ 
Also to show the effect of neighbor size selection, we choose $k$ from $\{3,5,10\}.$
We can see that when our method could control FDR under the presetted significant level $\alpha=0.05$ with small neighbor size $k=3,5;$ while FDR$_L$ would allow FDR a little bit higher than $0.05.$
When the neighbor size $k$ is small ($k=3$), SCUSUM outperforms both in false negative and false positive: though the false negative for SCUSUM is $0.6338$ when $k=3$ and $\mu=0.8,$ the corresponding false positive is $0.001$, a very small proportion and the false negative for FDR$_L$ under the same setting is $0.9831,$ almost $1.$
In addition, when we allow the neighbor size $k$ to be $5,$ imply a little larger block, the false negative for SCUSUM would decrease to be less than $0.50$ while the false positive is controlled less than $0.003.$
Compared with FDR$_L$, with neighbor size as $5,$ the false negative is over $0.50$ with low signal strength $0.8$ and $1.0,$ while the false positive is almost higher than $0.003.$
These experiment data implies that our proposed method, SCUSUM, performs bettern than FDR$_L$.

\begin{table}[t!]
\caption{Detected accuracy comparision between SCUSUM and FDR$_L$: the signal-noise ratio is changed from $0.8$ to $2$ and neighbor size is chose from $\{3,5,10\}$. In the two algorithms, significant level $\alpha$ is $0.05.$}\label{Table:Compare}
\small
\begin{center}
\begin{tabular}{c|c|c|c|c|c|c|c}
\hline
 \hline
 neighbor &\multicolumn{1}{c|}{Signal} & \multicolumn{3}{c|}{SCUSUM} & \multicolumn{3}{c}{FDR$_L$}\\
\hhline{~-------}
 size&$\mu$  & false negative & false positive & FDR & false negative & false positive & FDR\\
 \hline
  \hline
\multirow{4}{*}{ k=3 }& 0.8 &0.6338 & 0.0010 & 0.0186 & 0.9831 & 0.0003 & 0.0917 \\
\hhline{~-------}&1 & 0.4286 & 0.0010 & 0.0115 & 0.9380 & 0.0006 & 0.0484 \\
\hhline{~-------}&1.5&0.1953 & 0.0005 & 0.0044 & 0.5158 & 0.0039 & 0.0508 \\
\hhline{~-------}& 2&  0.1271 & 0.0006 & 0.0042 & 0.1478 & 0.0072 & 0.0535 \\
\hline
\hline
\multirow{4}{*}{ k=5 }&0.8&0.3750 & 0.0009 & 0.0094 & 0.8034 & 0.0020 & 0.0586 \\
\hhline{~-------}&1 & 0.2750 & 0.0009 & 0.0082 & 0.5351 & 0.0047 & 0.0612 \\
\hhline{~-------}&1.5&0.1599 & 0.0019 & 0.0147 & 0.0971 & 0.0112 & 0.0770 \\
\hhline{~-------}&2&  0.1014 & 0.0027 & 0.0194 & 0.0159 & 0.0150 & 0.0928 \\
\hline
\hline
\multirow{4}{*}{ k=10 }&0.8& 0.3240 & 0.0066 & 0.0588 & 0.1433 & 0.0286 & 0.1793\\
\hhline{~-------}&1 & 0.2526 & 0.0084 & 0.0692 & 0.0569 & 0.0394 & 0.2163 \\
\hhline{~-------}&1.5&0.1642 & 0.0133 & 0.0968 & 0.0076 & 0.0651 & 0.3051 \\
\hhline{~-------}&2&  0.1288 & 0.0149 & 0.1034 & 0.0027 & 0.0751 & 0.3358 \\
\hline
\hline
\end{tabular}
\end{center}
\end{table}

In Figure \ref{fig:compare_SpatialCUSUM_FDRl}, we show the probability maps of the pixels being detected by the two methods with give $\alpha=0.05$. 
Here we set the neighbor size $k=5,$ which we think has performance from Table \ref{Table:Compare} (relatively lower false positive and false negative). 
We range the signal strength from $0.5$ to $2.$ 
In the probability maps, the darker the color is, the higher probability the corresponding point is signal. 
Comparing the probability maps with the ground truth (see Figure \ref{fig:groundtruth}), it could be easily see that with hige signal strength $\mu\ge 1.5,$ the two methods have almost the same performance; while the signal is much too low (e.g.$\mu=0.5$,) SCUSUM has a better detected result than FDR$_L$.
Also we can see that in the results of FDR$_L$ there are some shadows outside of 'L' and 'H' signal region, which are false positive; while for SCUSUM, the detections for indifference region are more 'white' (no shadows.)
To some degree, these probability maps are consistent with the experiment data in Tabel \ref{Table:Compare}.

\begin{figure}[t!]
\centering
\includegraphics[width=15cm]{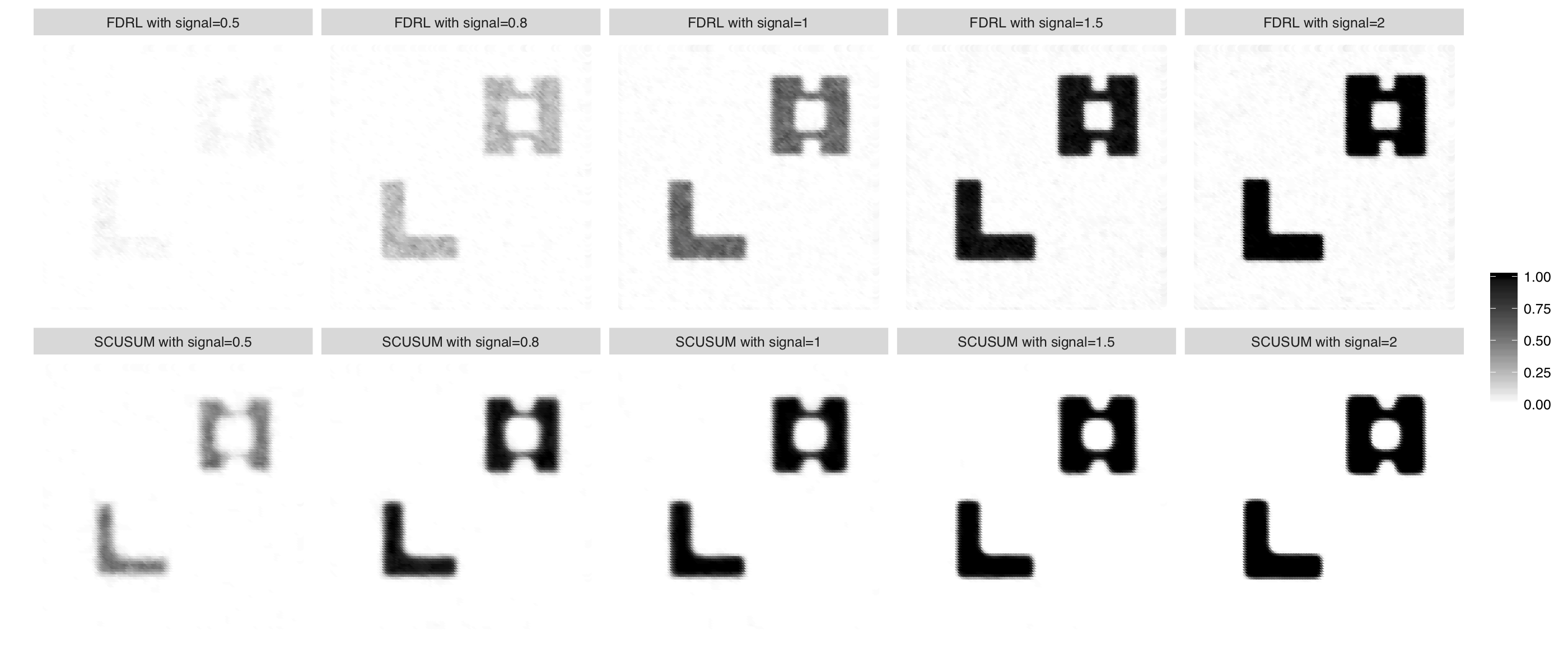}
 \caption{Comparision of the detection probability for SCUSUM and FDR$_L$ under different signal strength: the darker the color is, the higher probability the corresponding point is signal.}\label{fig:compare_SpatialCUSUM_FDRl}
\end{figure}

Though in model \ref{basic_model} we assume the noise to be independence, now we try to apply SCUSUM to dependence spatial observations, and compare with FDR$_L$.
Here we consider to use the Exponential Covariance Model (\cite{gelfand2010handbook}) to generate dependence data. 
The covariance matrix is 
\begin{equation}\label{dep_cov}
C(s_i,s_j)=\exp(-\frac{\|s_i-s_j\|}{r}), 
\end{equation}
where $s_i$ and $s_j$ are two locations, $\|s_i-s_j\|$ is the distance between the two location, and $r$ is the dependence scale. 
The larger $r$ means the stronger dependence.
Then the data are generate from multivariate normal distribution with above covariance matrix and correponding mean, $\mu_0$ for $\mathscr{D}_{\mathscr{A}^c}$ and $\mu_1$ for $\mathscr{D}_{\mathscr{A}}.$
We range scale $r$ from $\{0.1,0.3,0.5\},$ and the corresponding covariances for unit distance are $\{0.00004,0.03567,0.13533\}.$
We set the neighbor size $k=5$ for both two alogrithms and $\alpha=0.05$ to control marginal FDR.
The results are shown in Figure \ref{fig:compare_SpatialCUSUM_FDRl_dep} and Tabel \ref{Table:dependence}.
It could be seen that with the weak dependence, SCUSUM could still detect the signal region efficiently while control the false posive.
However for FDR$_L,$ the false negative increases significantly with stronger dependence (see Table \ref{Table:dependence}).
Also Table \ref{Table:dependence} shows that larger dependence scale leads to larger false positive, false negative and FDR for SCUSUM.
This gives us a hint that for larger scale dependence noise we need to choose relatively larger blocks.
Nevertheless, we can see that SCUSUM could recognize the signal region with higher probability than FDR$_L$, when the noise dependence is weak.

\begin{figure}[t!]
\centering
\includegraphics[width=12cm]{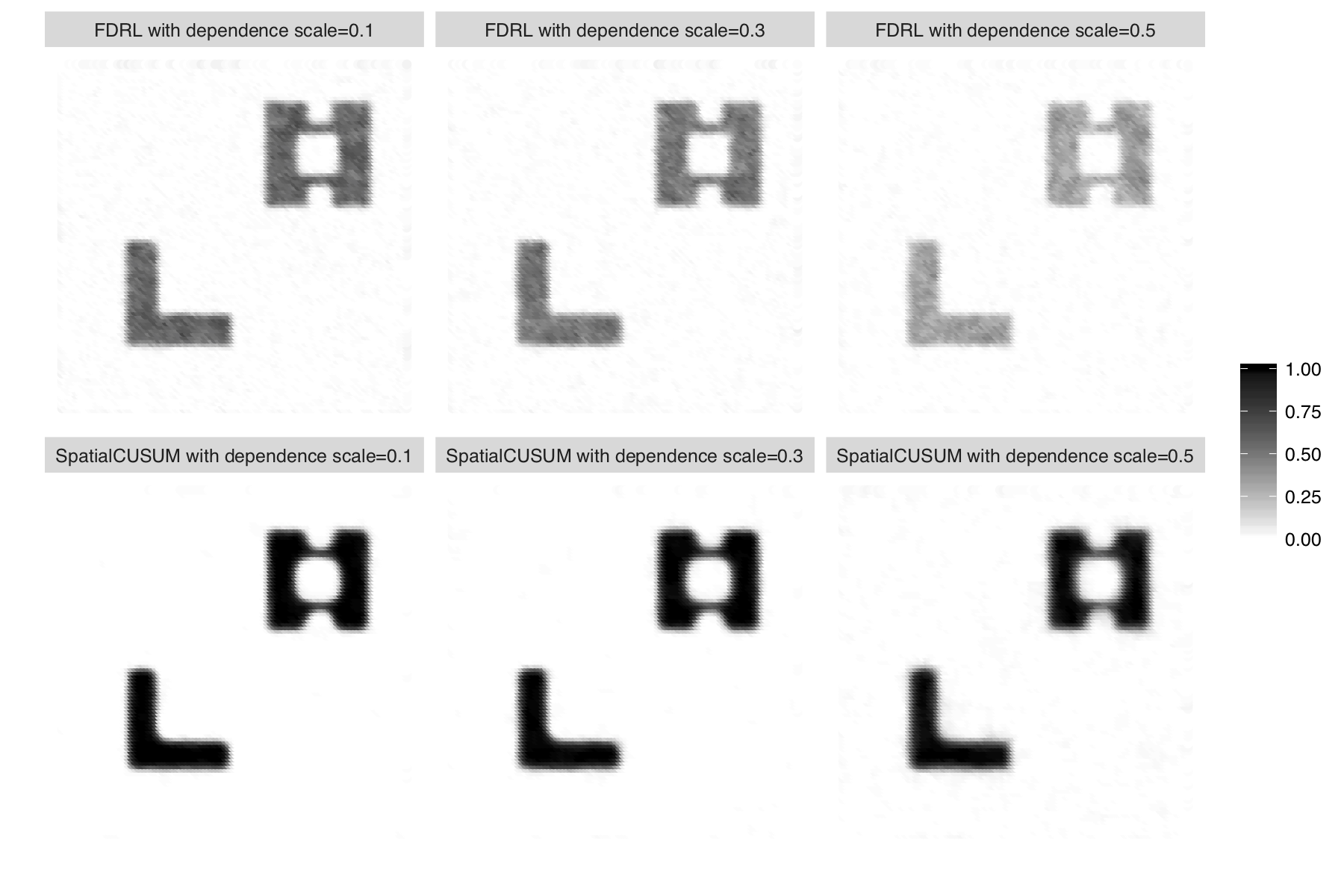}
 \caption{Comparision of the detection probability for SCUSUM and FDR$_L$ on data with different dependence scales: the darker the color is, the higher probability the corresponding point is signal. Here we set the neighbor size $k=5$ and signal $\mu=1.$}\label{fig:compare_SpatialCUSUM_FDRl_dep}
\end{figure}

\begin{table}[t!]
\caption{Summary for detection probabilities on dependence data. The dependence scale is changed from $0.1$ to $0.5,$ neighbor size is chose as $5$ and signal strength is $1.$}\label{Table:dependence}
\begin{center}
\begin{tabular}{c|c|c|c|c}
\hline
\hline
 \multicolumn{2}{c|}{Dependence Scale}  & $r=0.1$ & $r=0.3$ & $r=0.5$ \\
\hline
\hline
\multirow{2}{*}{false negative}    &SCUSUM &0.2721  & 0.2820 &  0.3108  \\
\hhline{~----}
       & FDR$_L$ & 0.5213 & 0.5879 & 0.7317 \\
\hline
\hline
\multirow{2}{*}{false positive}   &SCUSUM & 0.00095 &  0.00154 &  0.00418  \\
\hhline{~----}
& FDR$_L$ &  0.00499 & 0.00411 & 0.00280  \\
\hline
\hline
\multirow{2}{*}{ FDR} &SCUSUM &  0.0086 &  0.0141 &  0.0386  \\
\hhline{~----}
                               & FDR$_L$ & 0.0622 & 0.0604 & 0.0590 \\
\hline
\hline
\end{tabular}
\end{center}
\end{table}

\section{Real Data Experiment} \label{sec:Experiment}

In this section, we apply four methods, SCUSUM, scan statistics, the conventional FDR and FDR$_L$, to a real fMRI data to illustrate their differences in real data application. The fMRI data has been analyzed in some previous works (\cite{maitra2009assessing}; \cite{zhang2012spatial} etc.)

Figure \ref{fMRI} (a) shows six slices of the fMRI images in a total of 22 slices. Each individual image has $128\times128$ pixels. 
All these images show activities in different regions by heat maps. 
The pixels' values are the transformations of $p$-values from a previous study, which should follow a standard normal distribution, and we only care about detection of the regions with positive values.

As for the conventional FDR, the lack of identification phenomenon happened, e.g. all the signals of the fifth slice in Figure \ref{fMRI} (b) are missed .
Also without considering spatial correlation, for example in the first slice of Figure \ref{fMRI} (b), some detected 'signals' are scattered around, which means that some of them might be false positive. 
Although FDR$_L$ could make full use of the neighbor information of spatially structured data and improve the detection efficiency, many weak signals are missed, e.g. in fifth slice and sixth slice of Figure \ref{fMRI} (c), many active regions are not detected.
For scan statistics, even though almost all the 'hot' pixels are detected, the 'signal' regions are too large, which is doubtable. 

Similar to the conclusion in Section \ref{sec:SimuStudy}, SCUSUM is more likely to detect weak signals compared with FDR methods. 
In all the six slices, the detected regions are larger and cover the regions detected by FDR methods.
Meanwhile, within each slice, SCUSUM could identifiy several irregularly shaped clusters.
We can see that the detected regions form natrual clusters, and they are the spatially grouped 'hot' parts in the raw images.
These results show that our proposed method might be more suitable for detection of irregular shaped and weak spatial signals.

\begin{figure}[t!]
\centering
\begin{tabular}{@{}ccccc@{}}
(a) & (b) & (c) & (d) & (e)\\
\includegraphics[width=3cm]{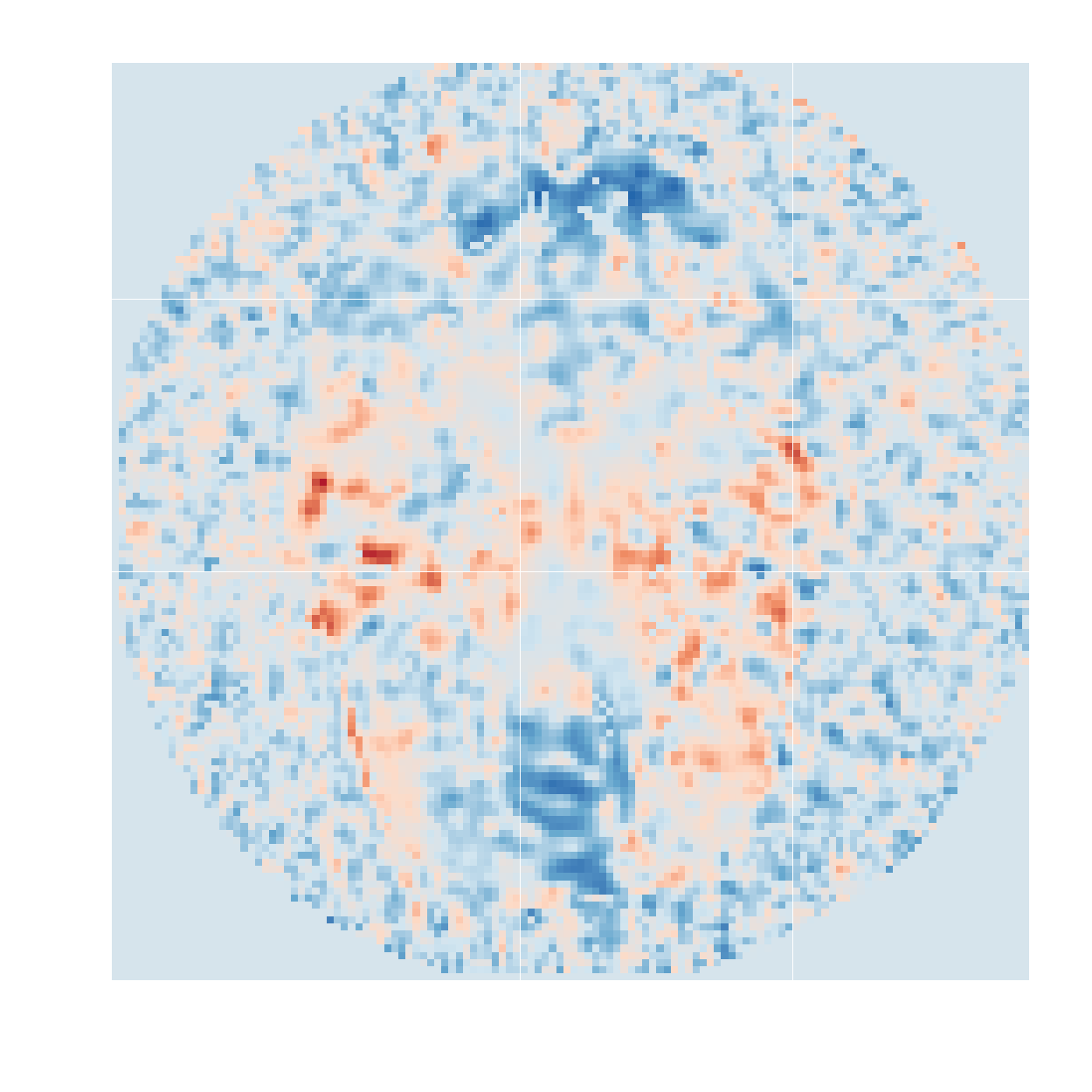}&
\includegraphics[width=3cm]{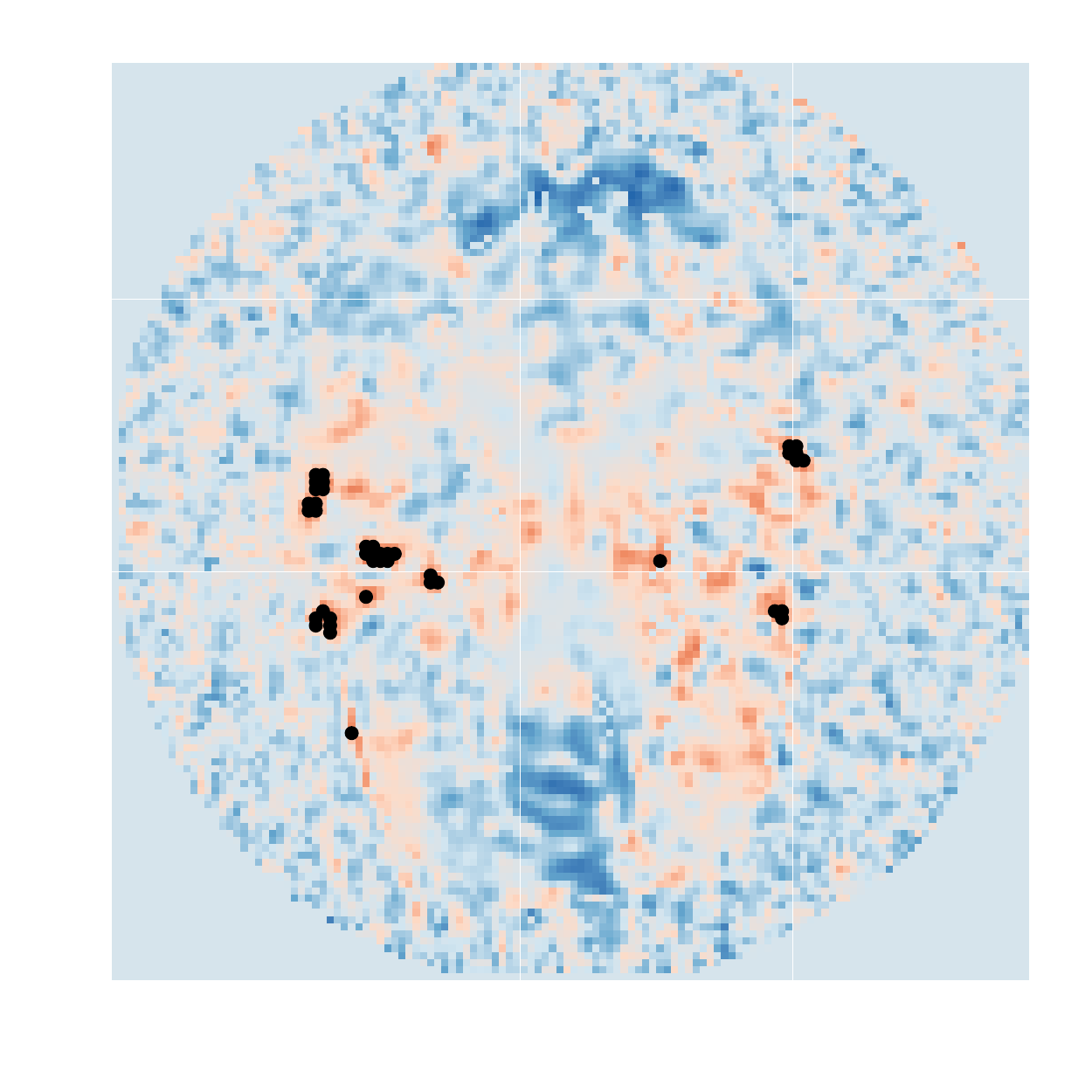}&
\includegraphics[width=3cm]{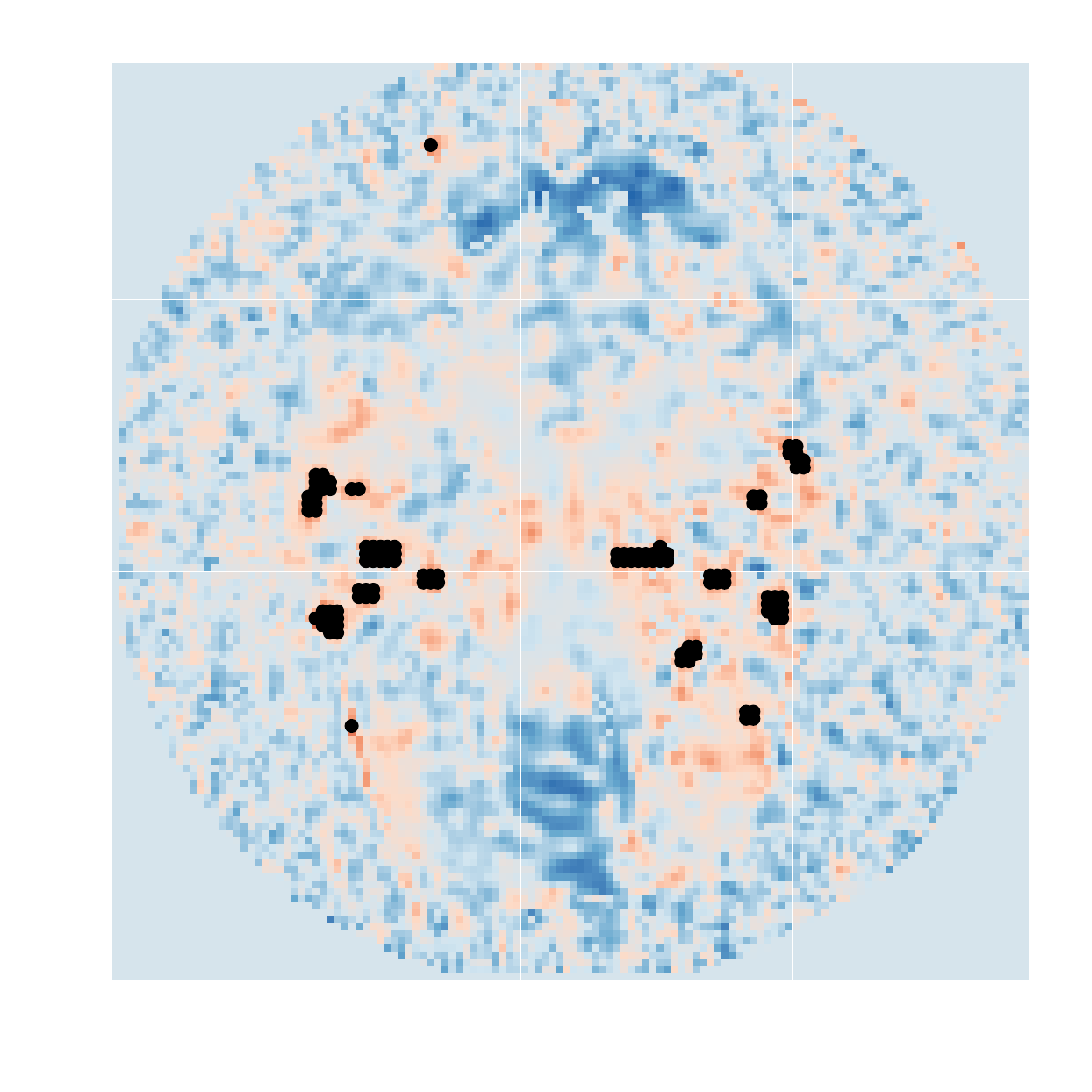}&
\includegraphics[width=3cm]{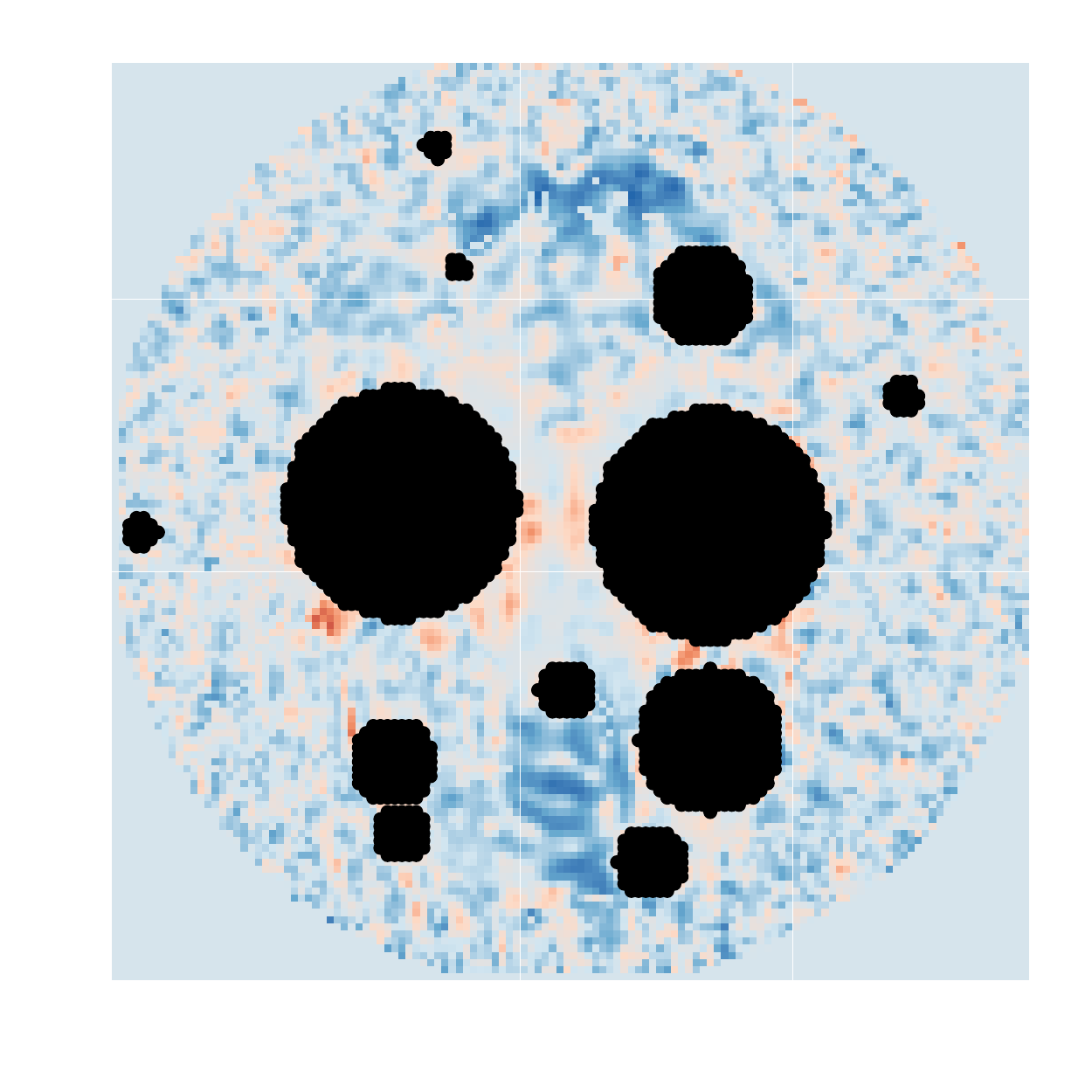}&
\includegraphics[width=3.7cm]{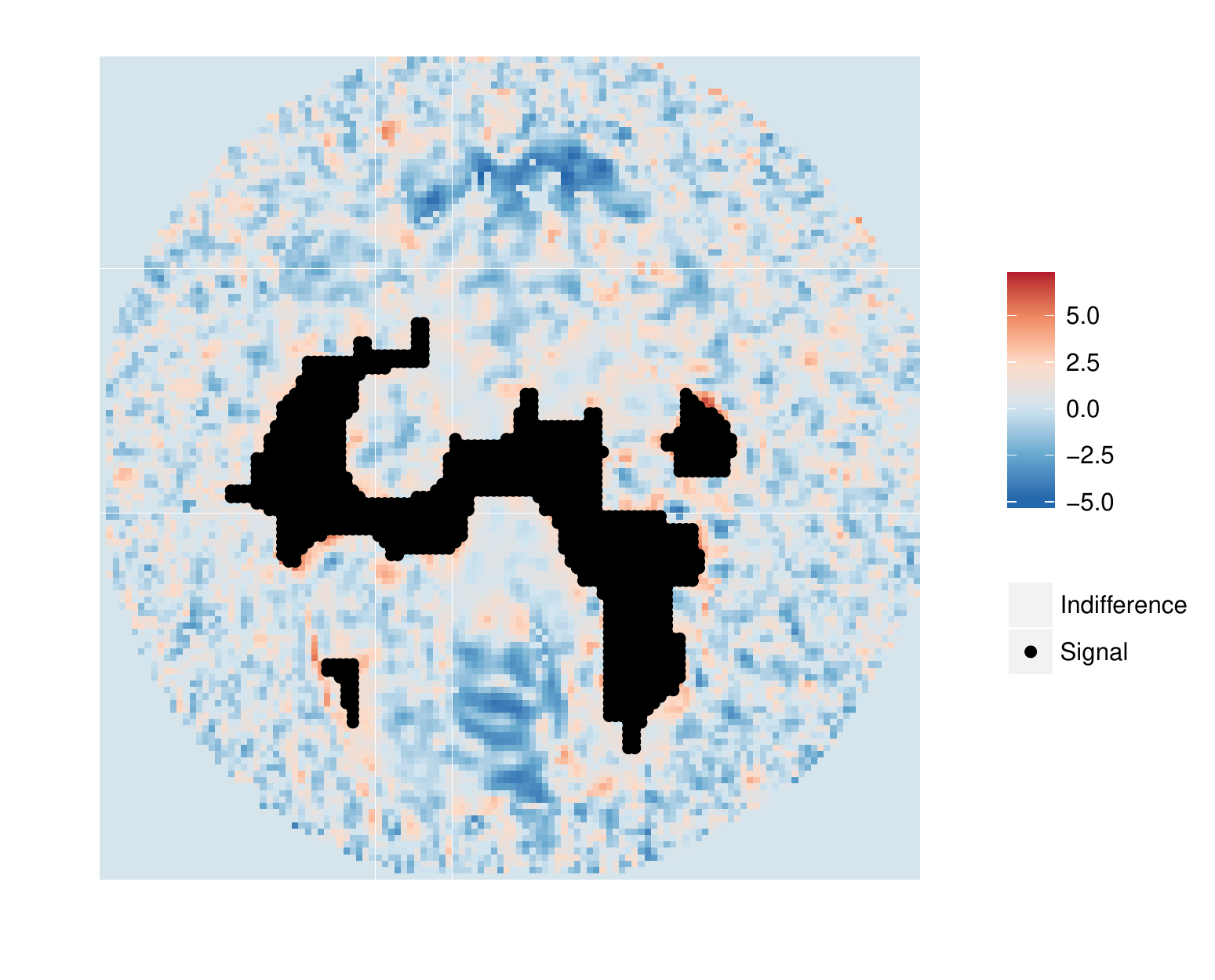}\\
\includegraphics[width=3cm]{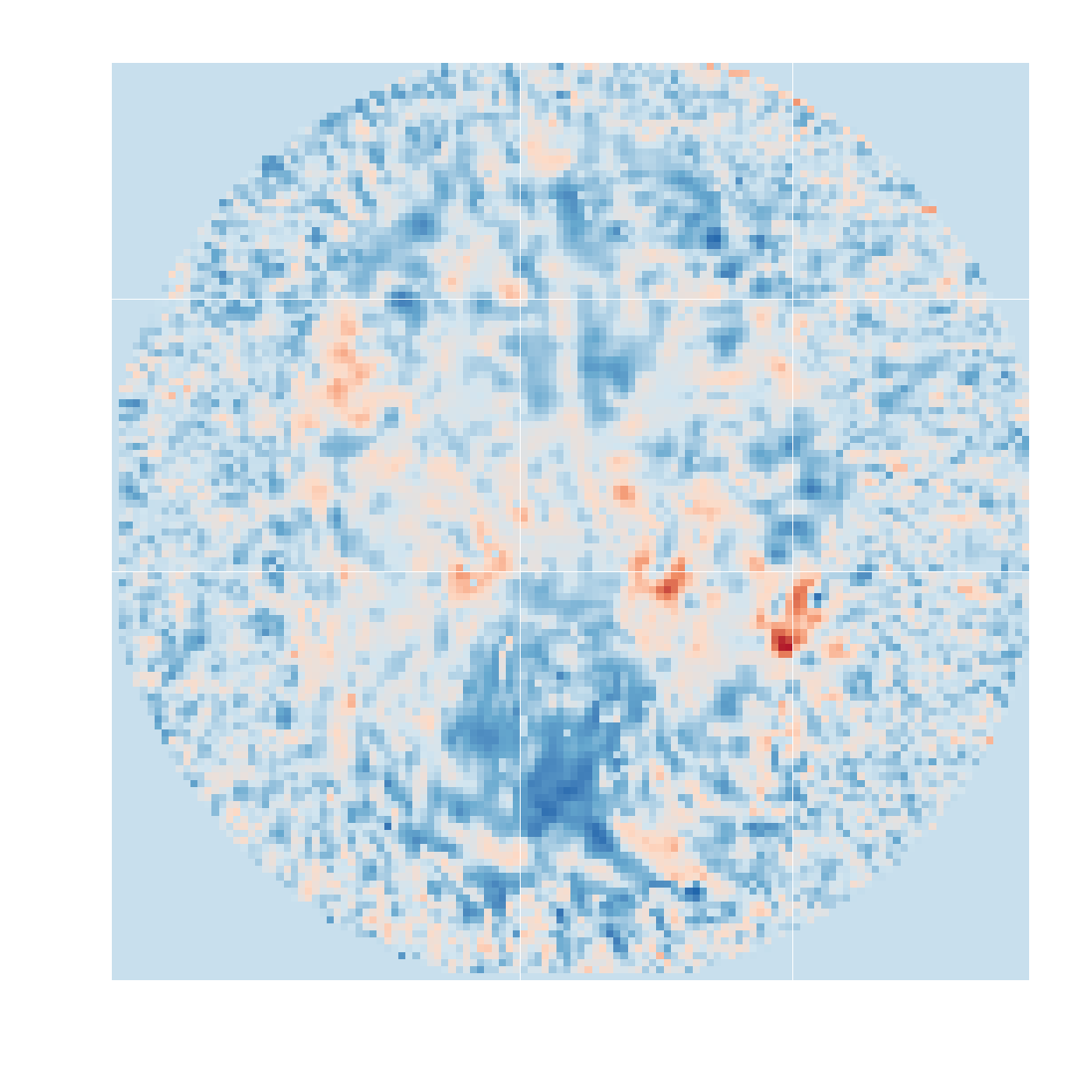}&
\includegraphics[width=3cm]{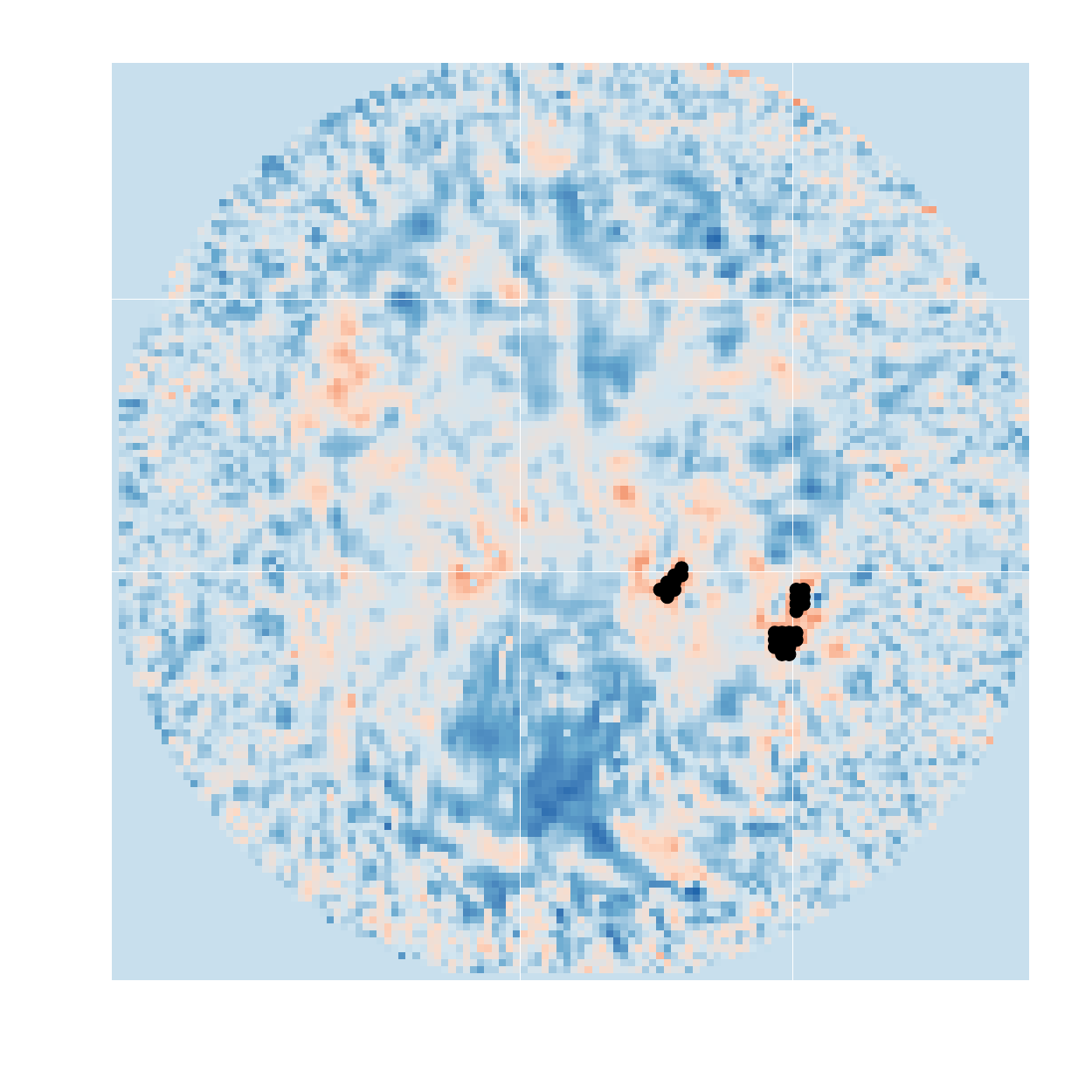}&
\includegraphics[width=3cm]{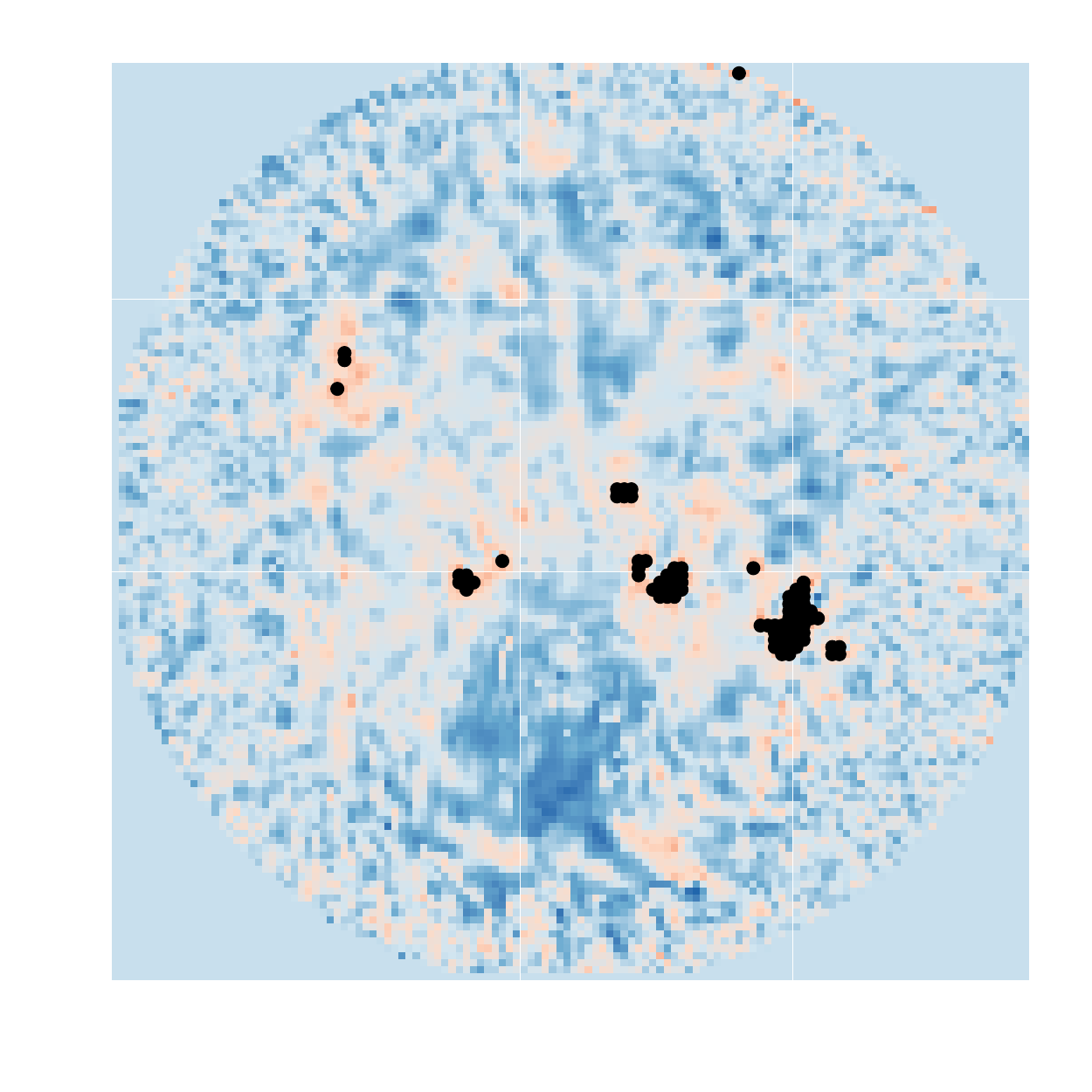}&
\includegraphics[width=3cm]{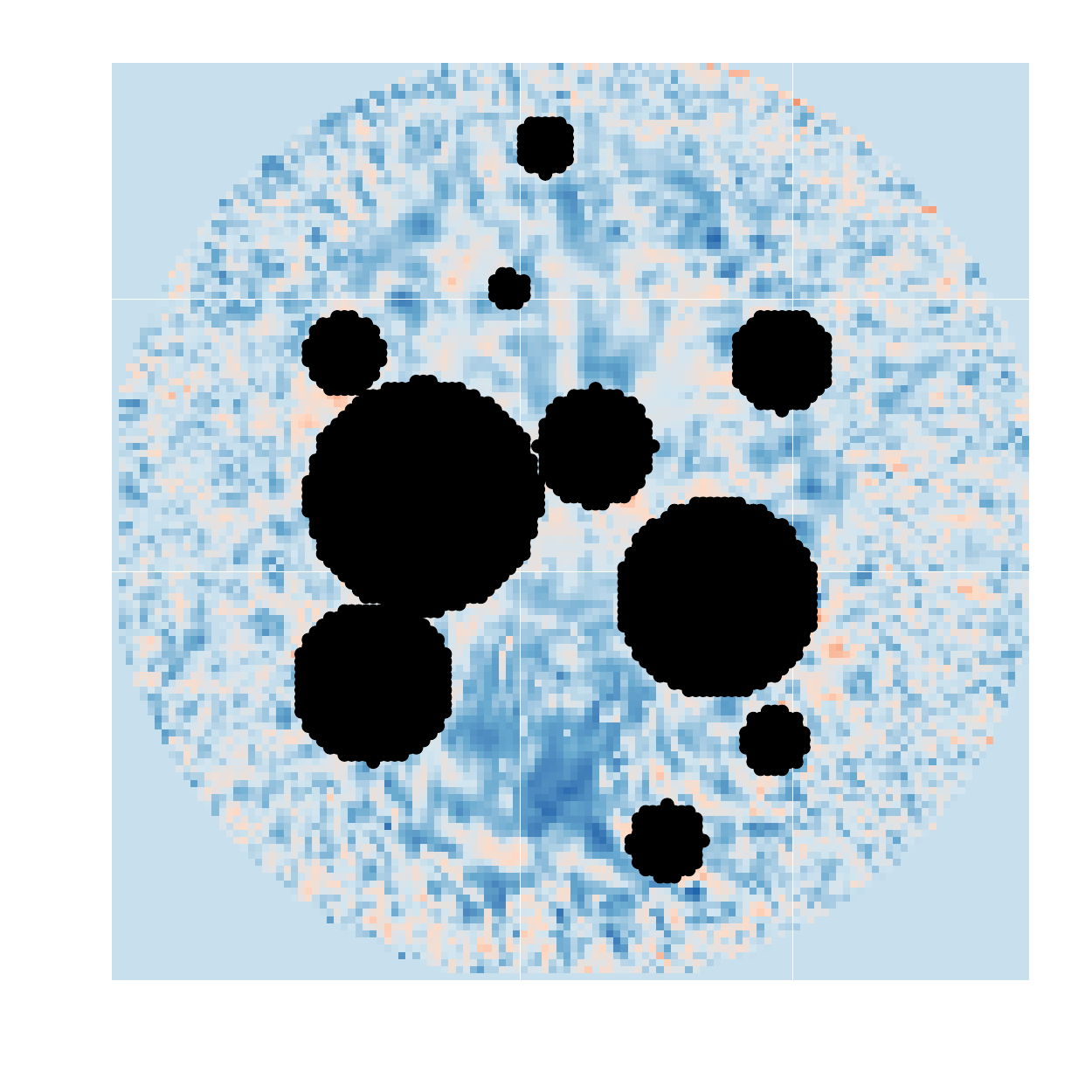}&
\includegraphics[width=3.7cm]{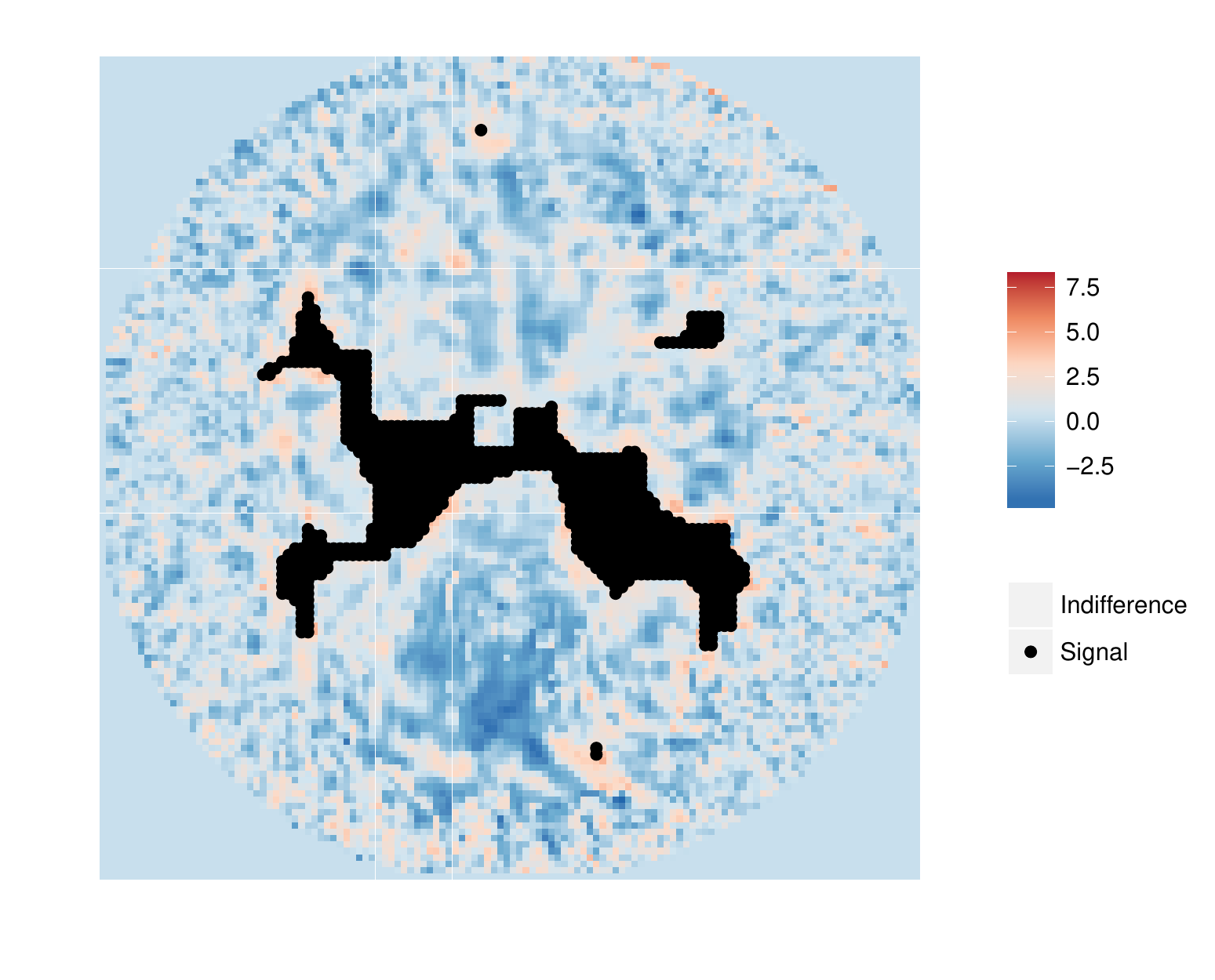}\\
\includegraphics[width=3cm]{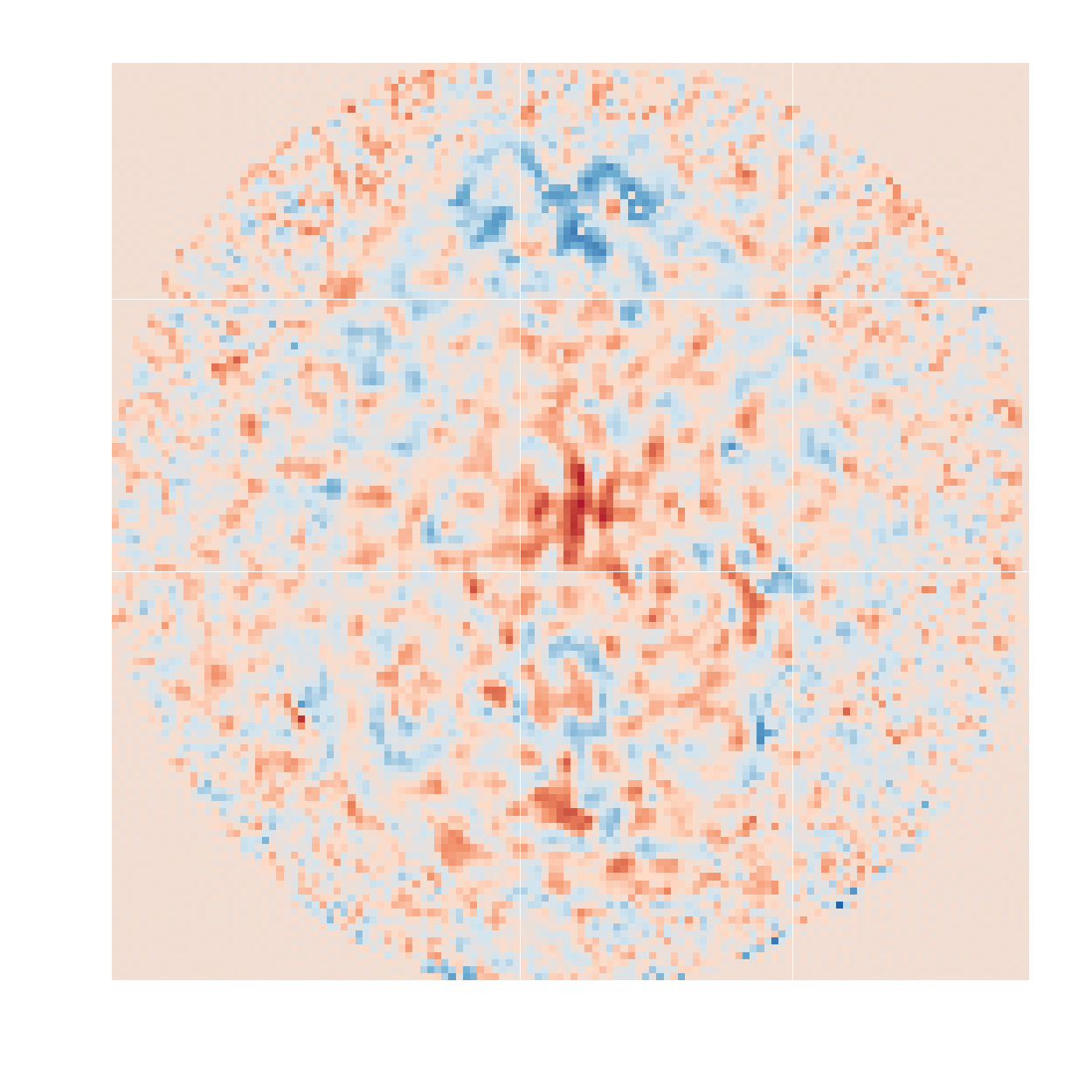}&
\includegraphics[width=3cm]{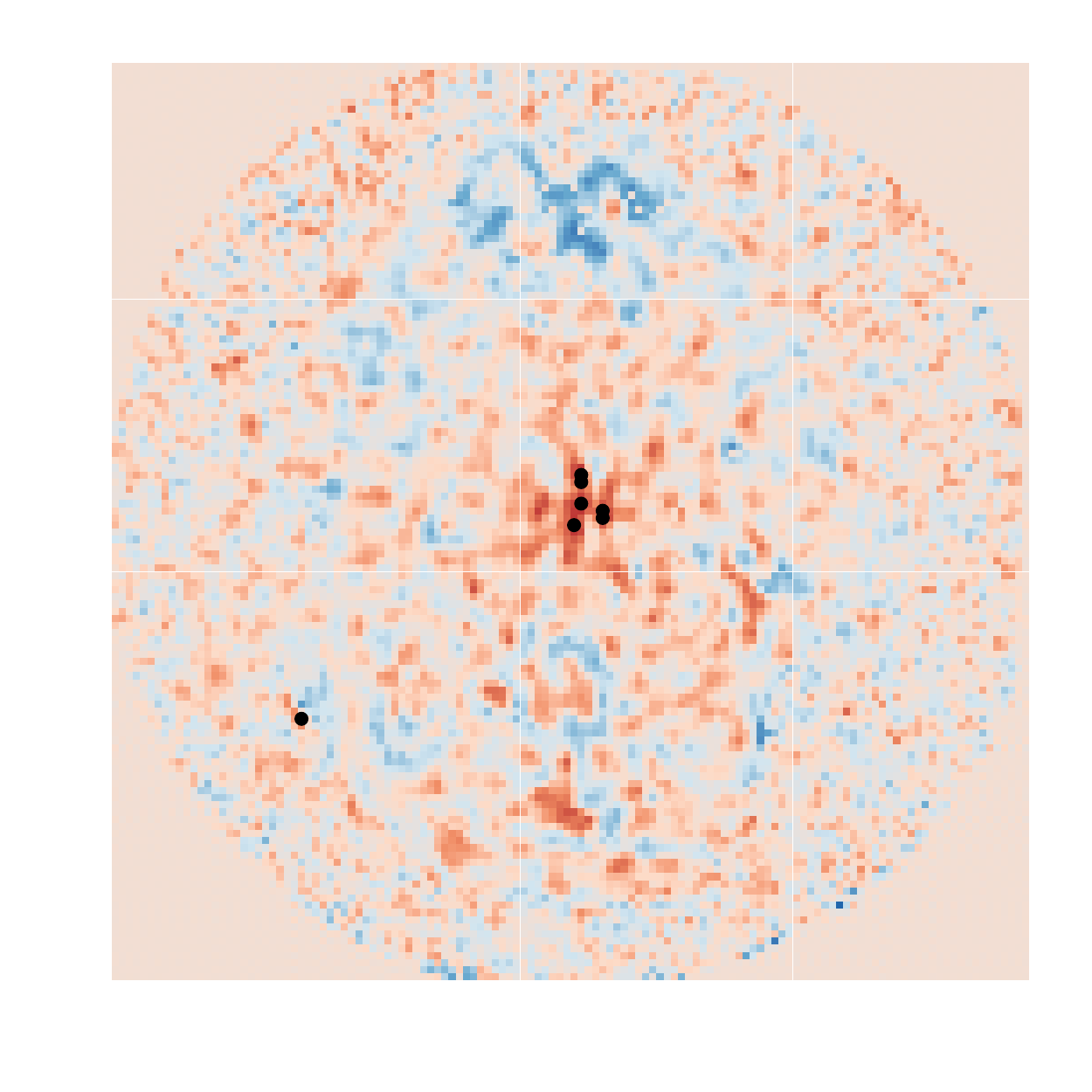}&
\includegraphics[width=3cm]{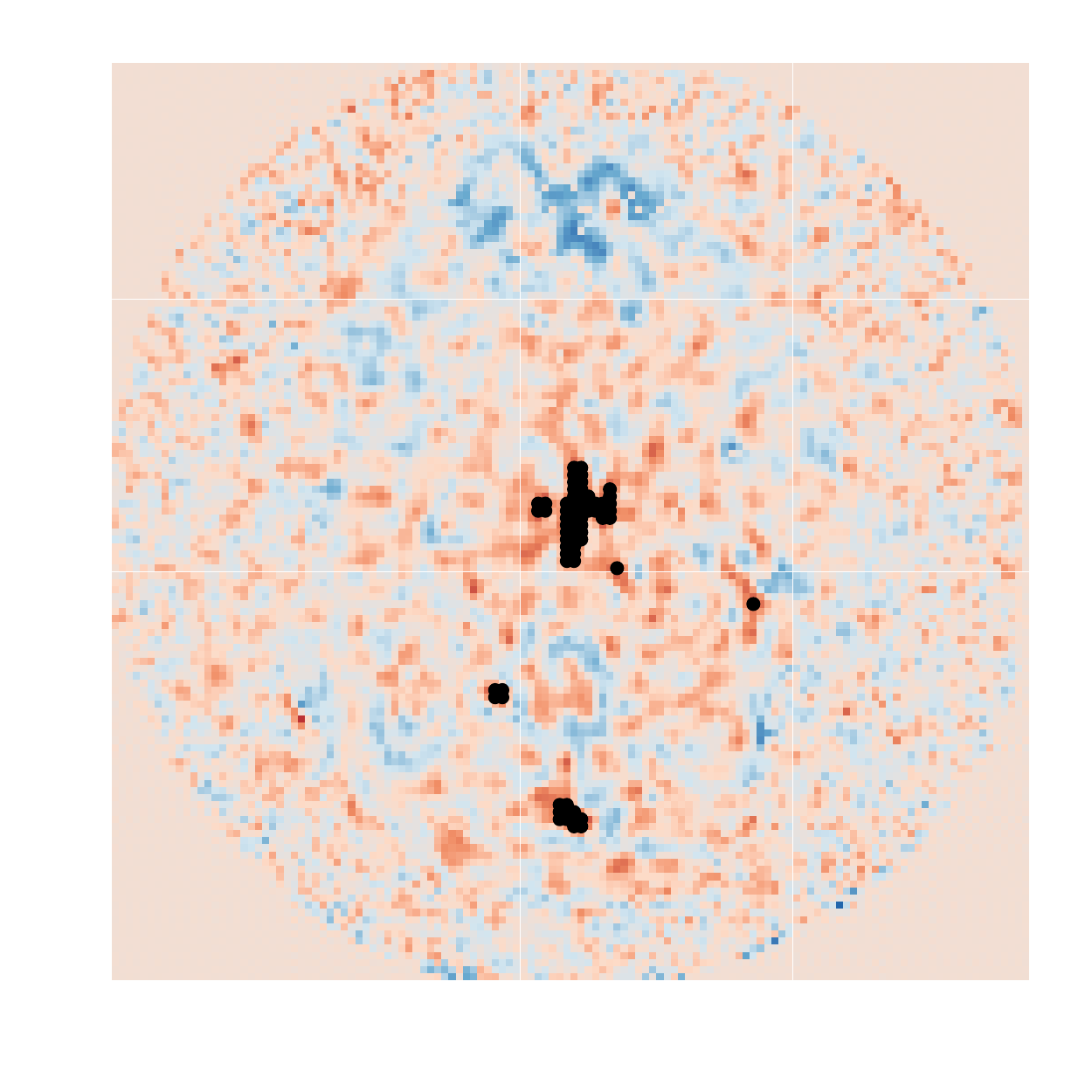}&
\includegraphics[width=3cm]{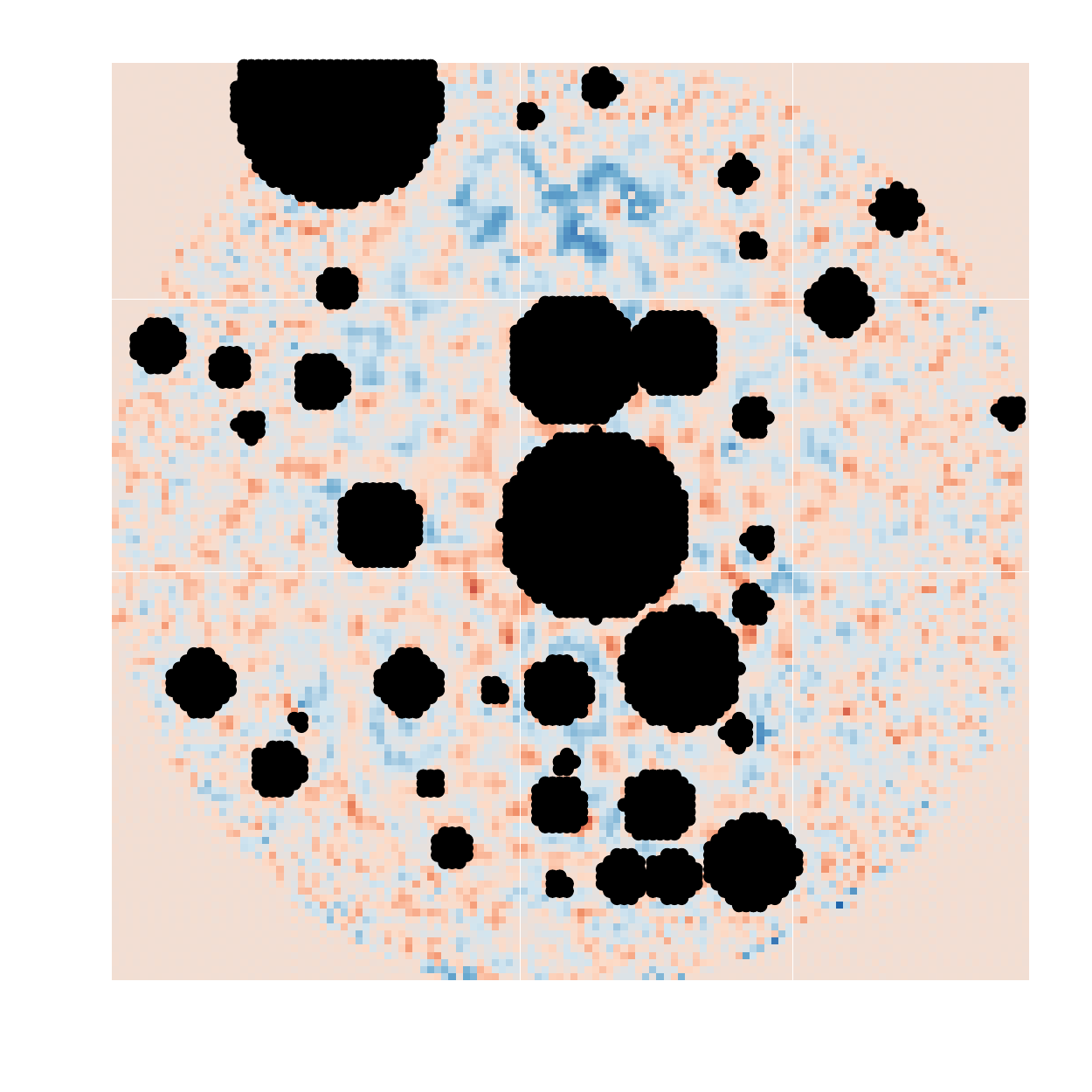}&
\includegraphics[width=3.7cm]{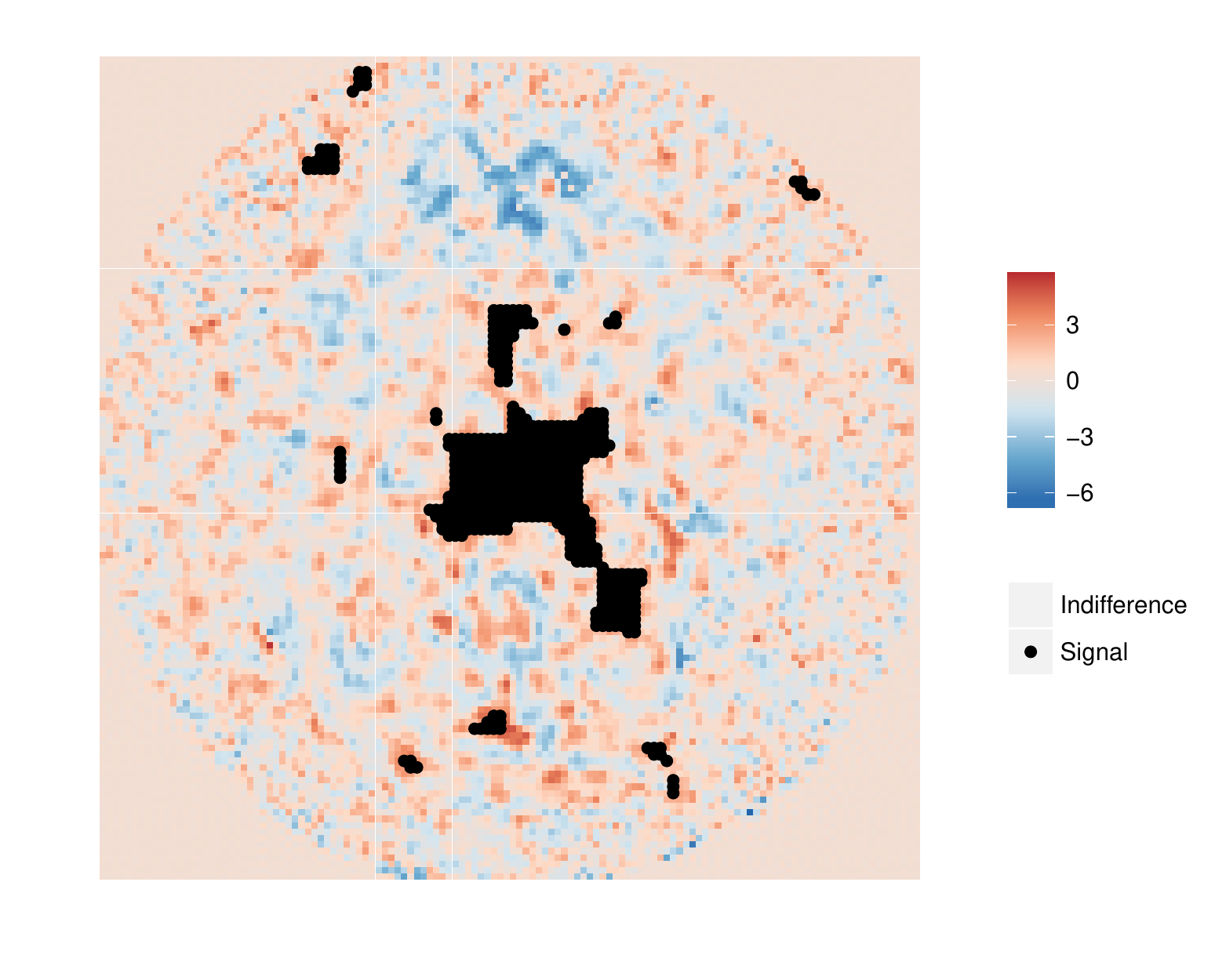}\\
\includegraphics[width=3cm]{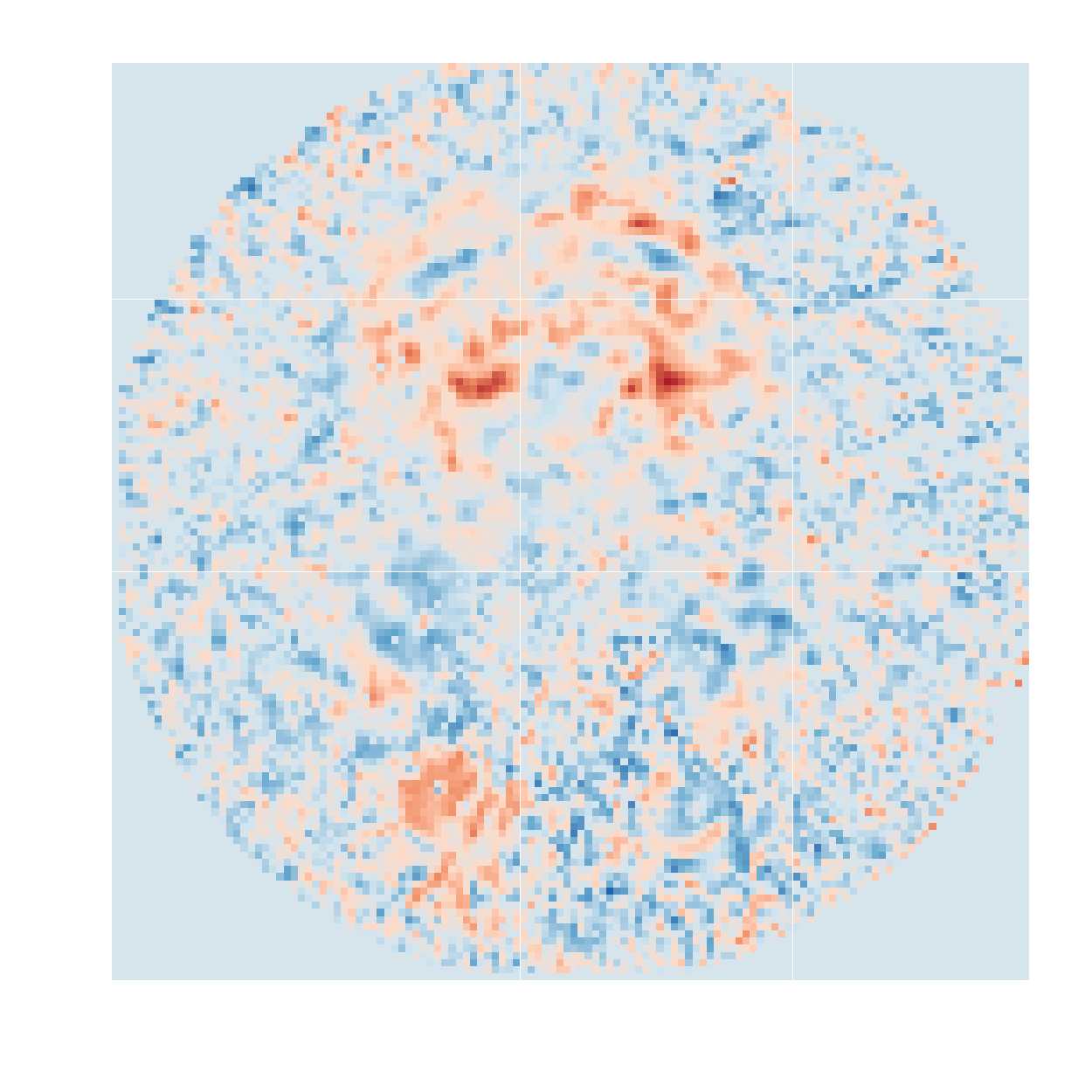}&
\includegraphics[width=3cm]{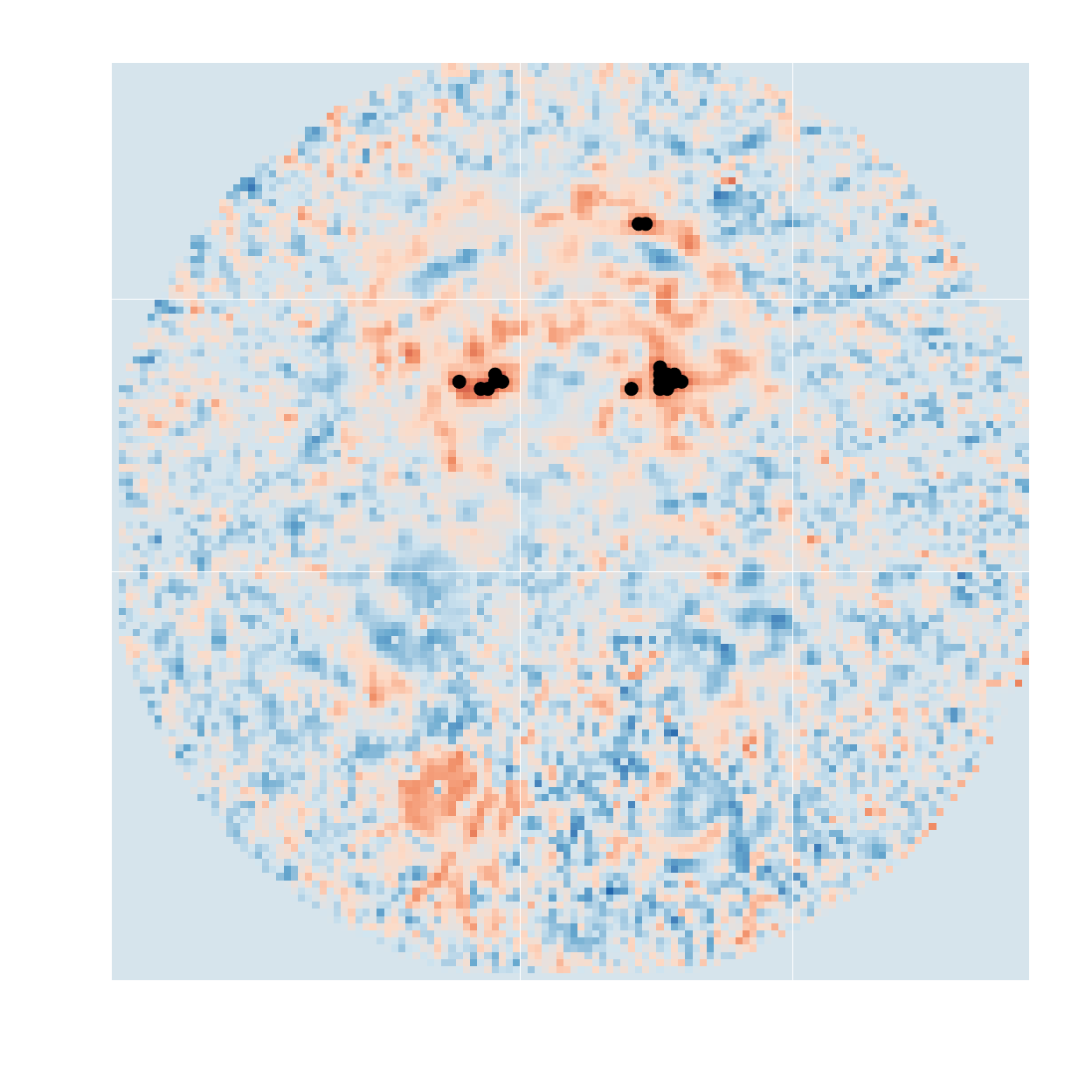}&
\includegraphics[width=3cm]{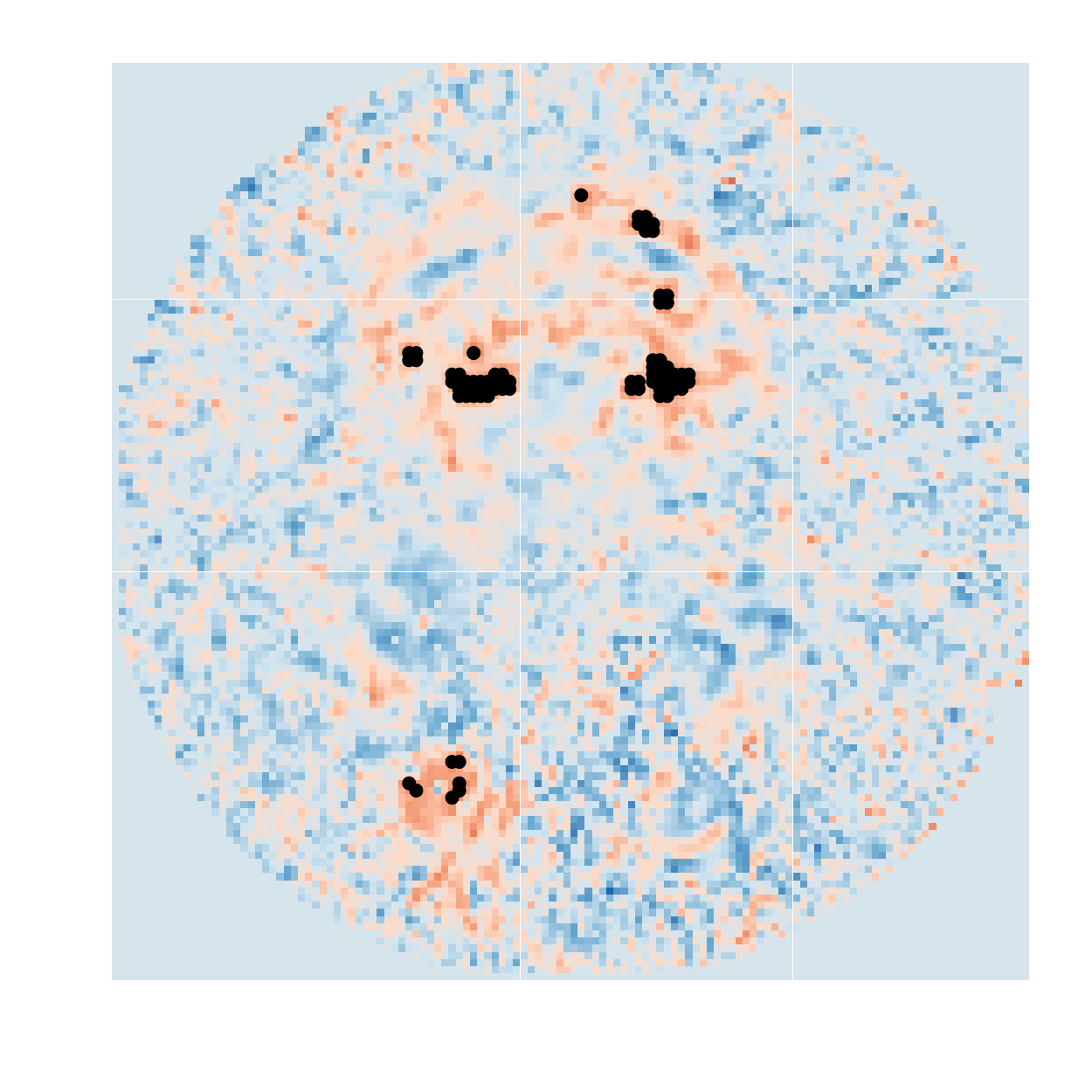}&
\includegraphics[width=3cm]{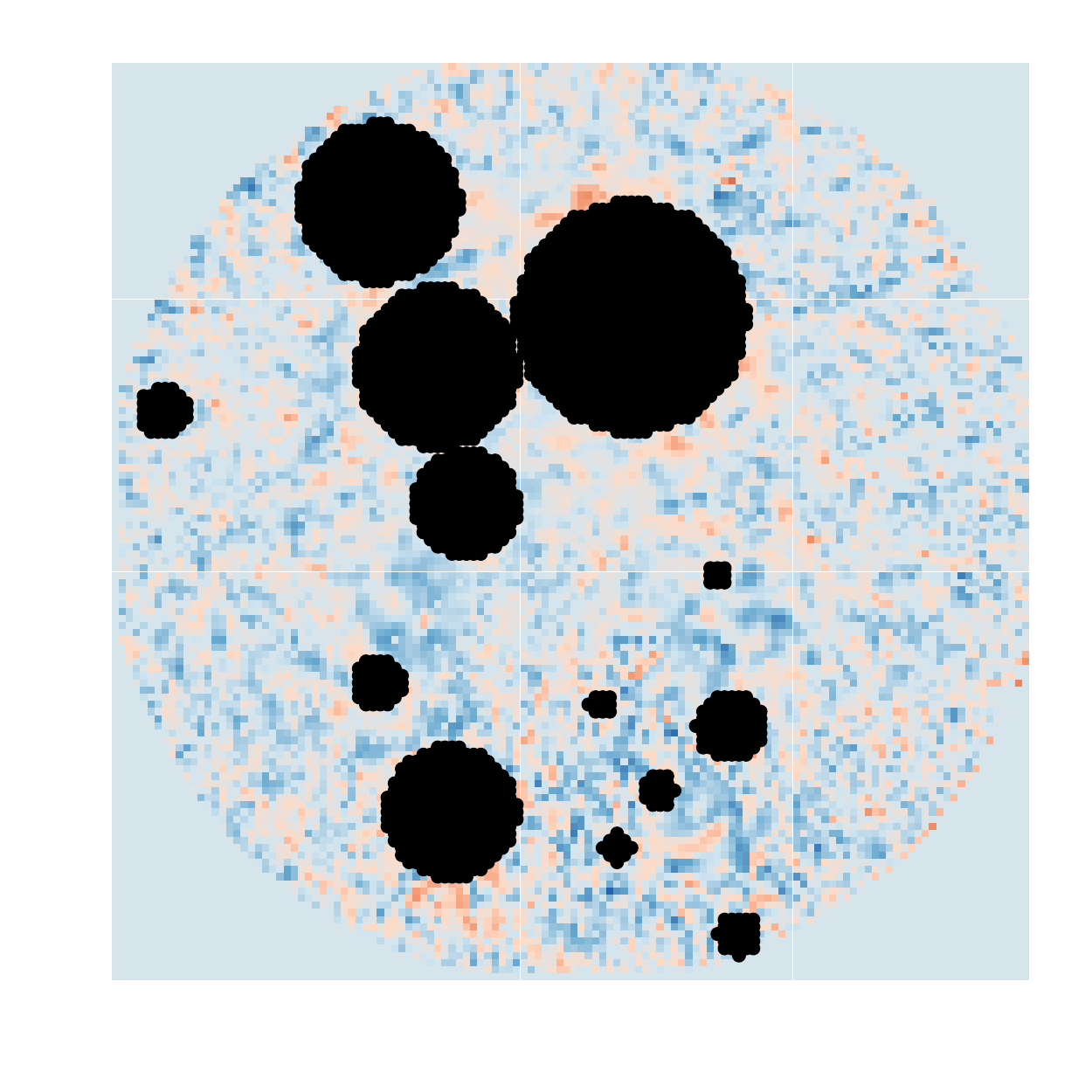}&
\includegraphics[width=3.7cm]{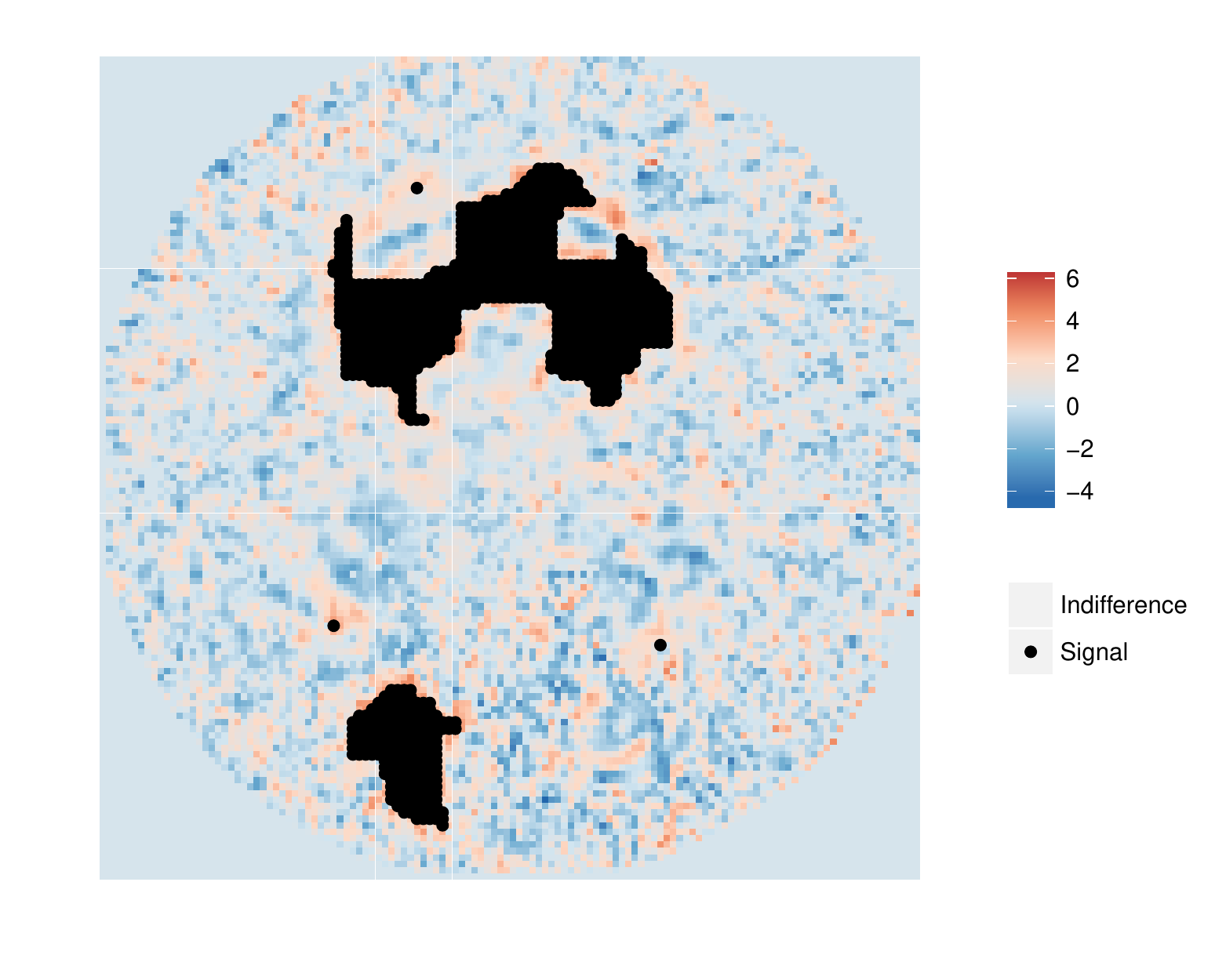}\\
\includegraphics[width=3cm]{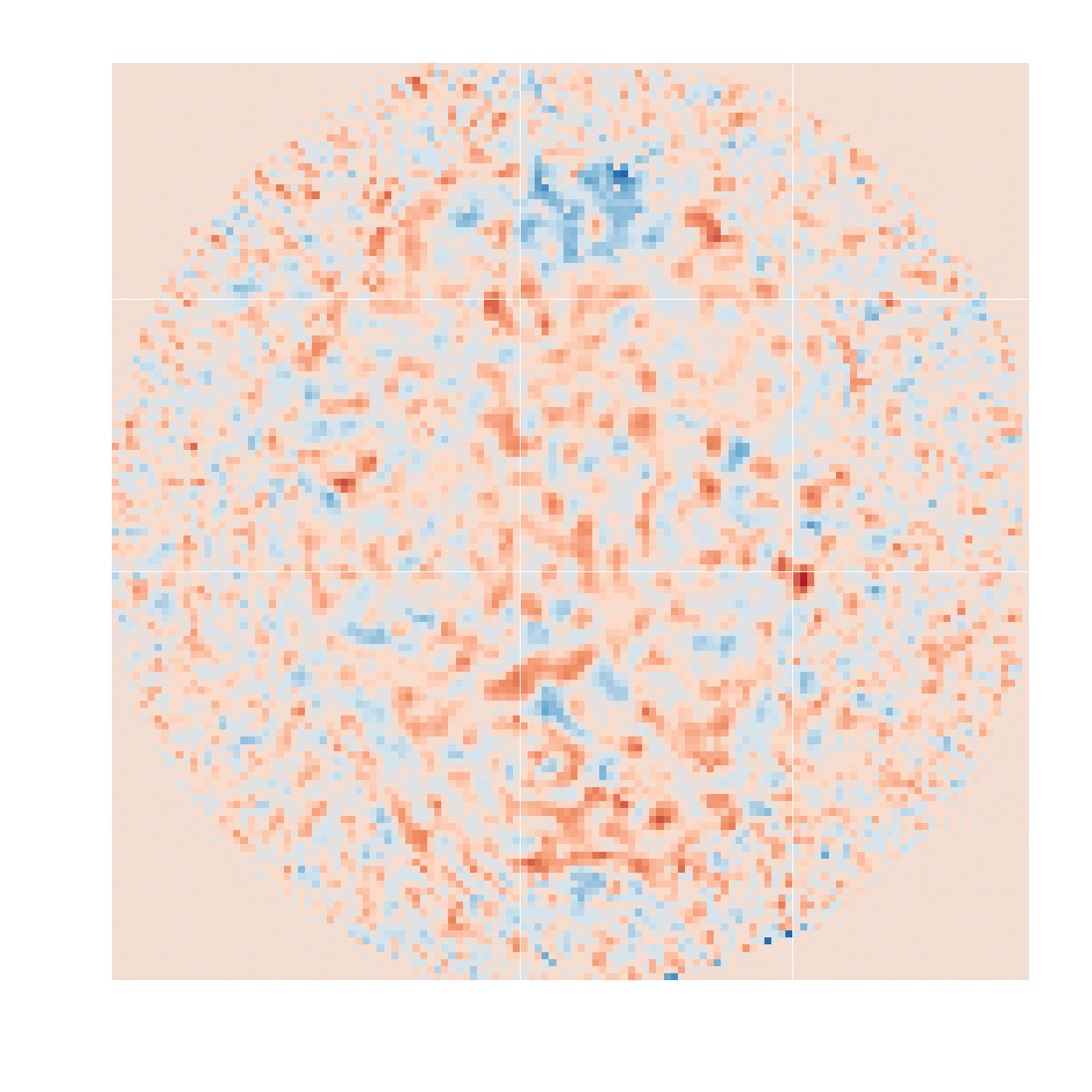}&
\includegraphics[width=3cm]{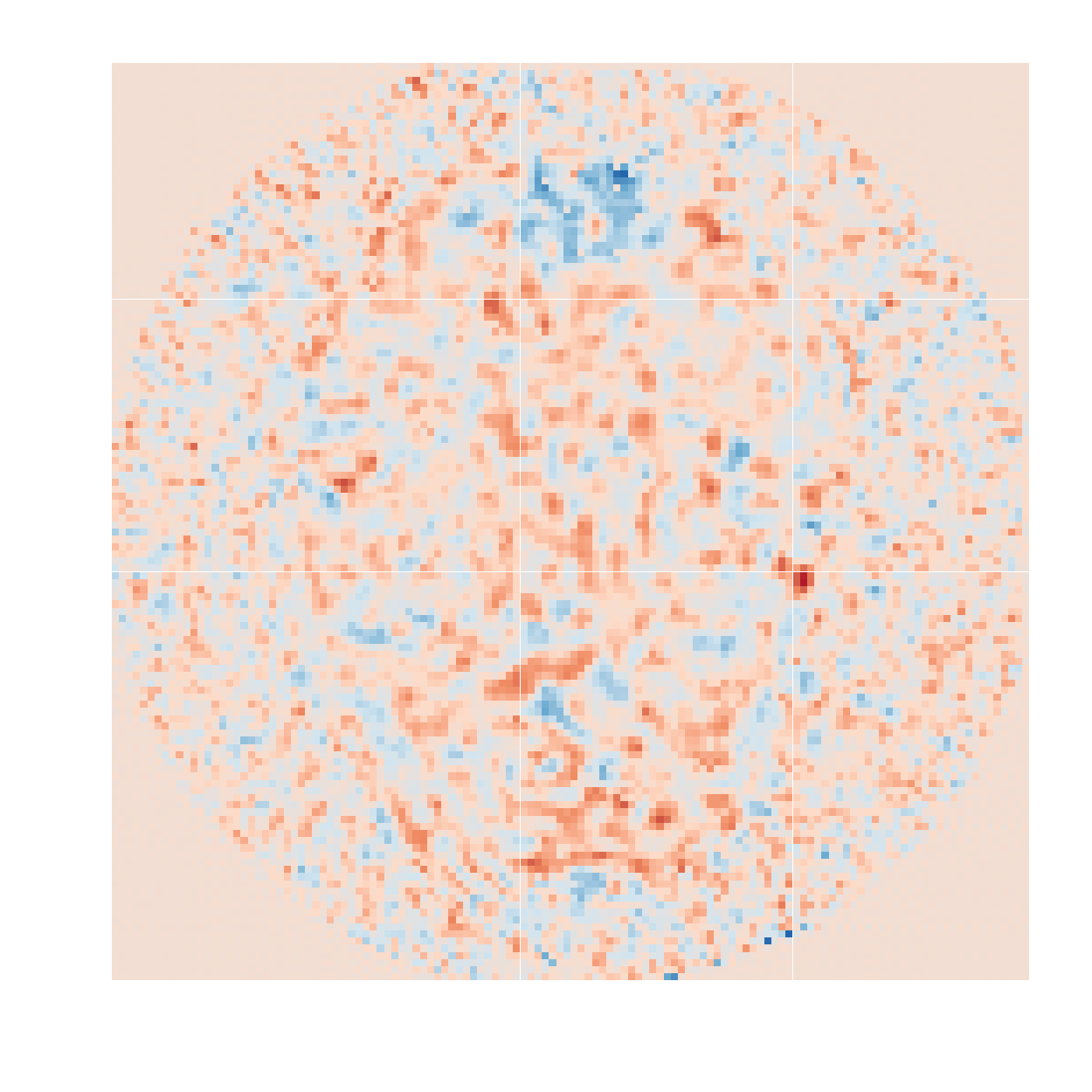}&
\includegraphics[width=3cm]{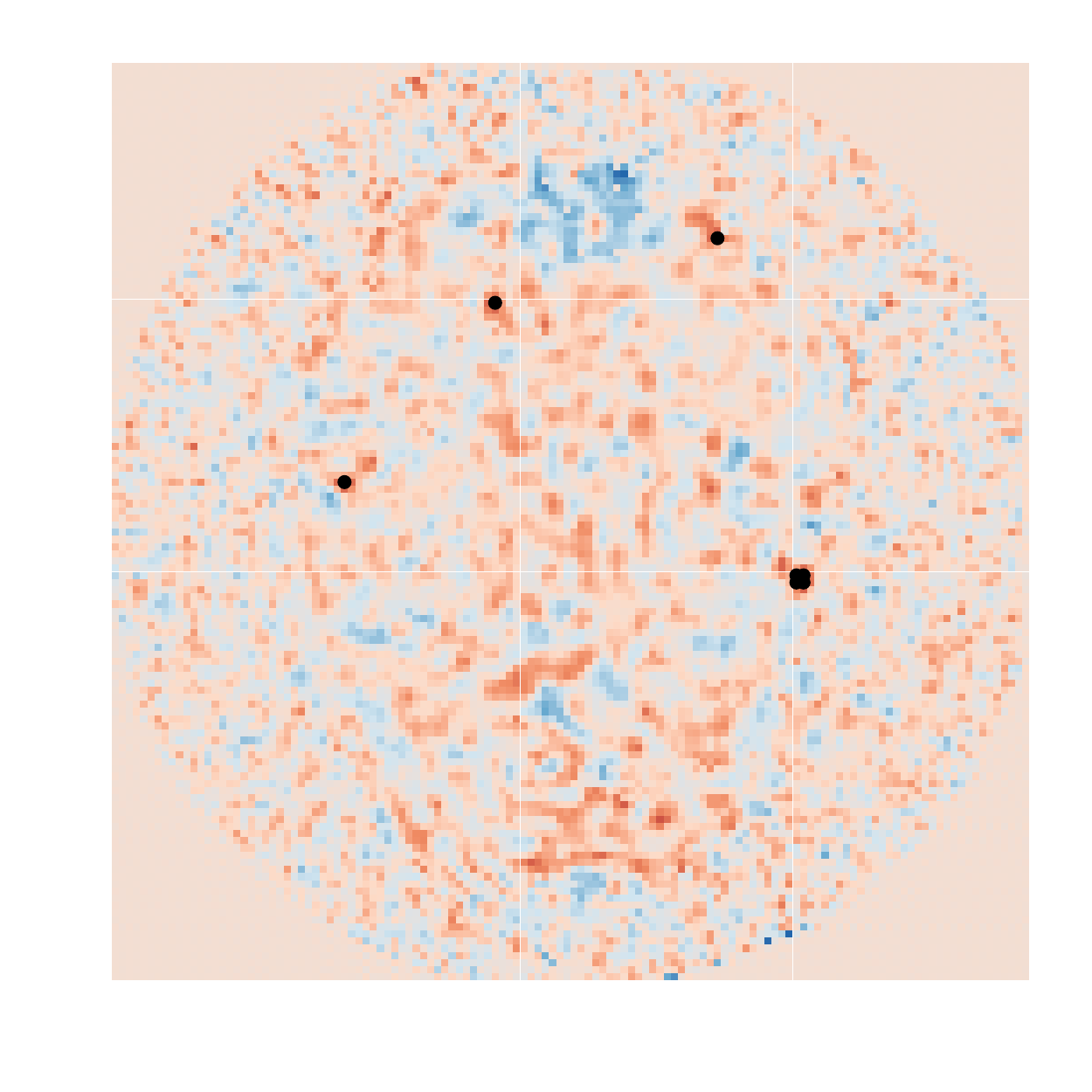}&
\includegraphics[width=3cm]{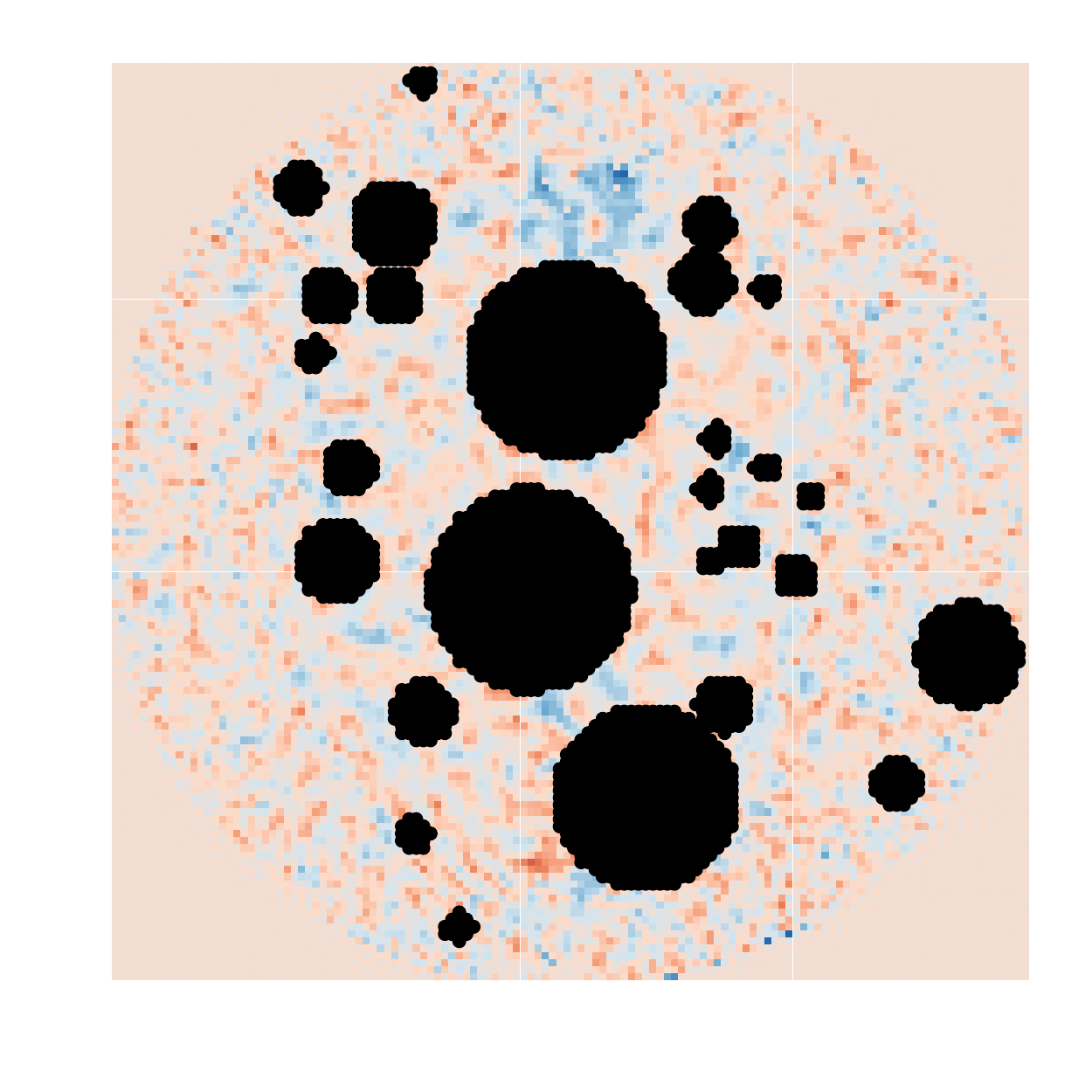}&
\includegraphics[width=3.7cm]{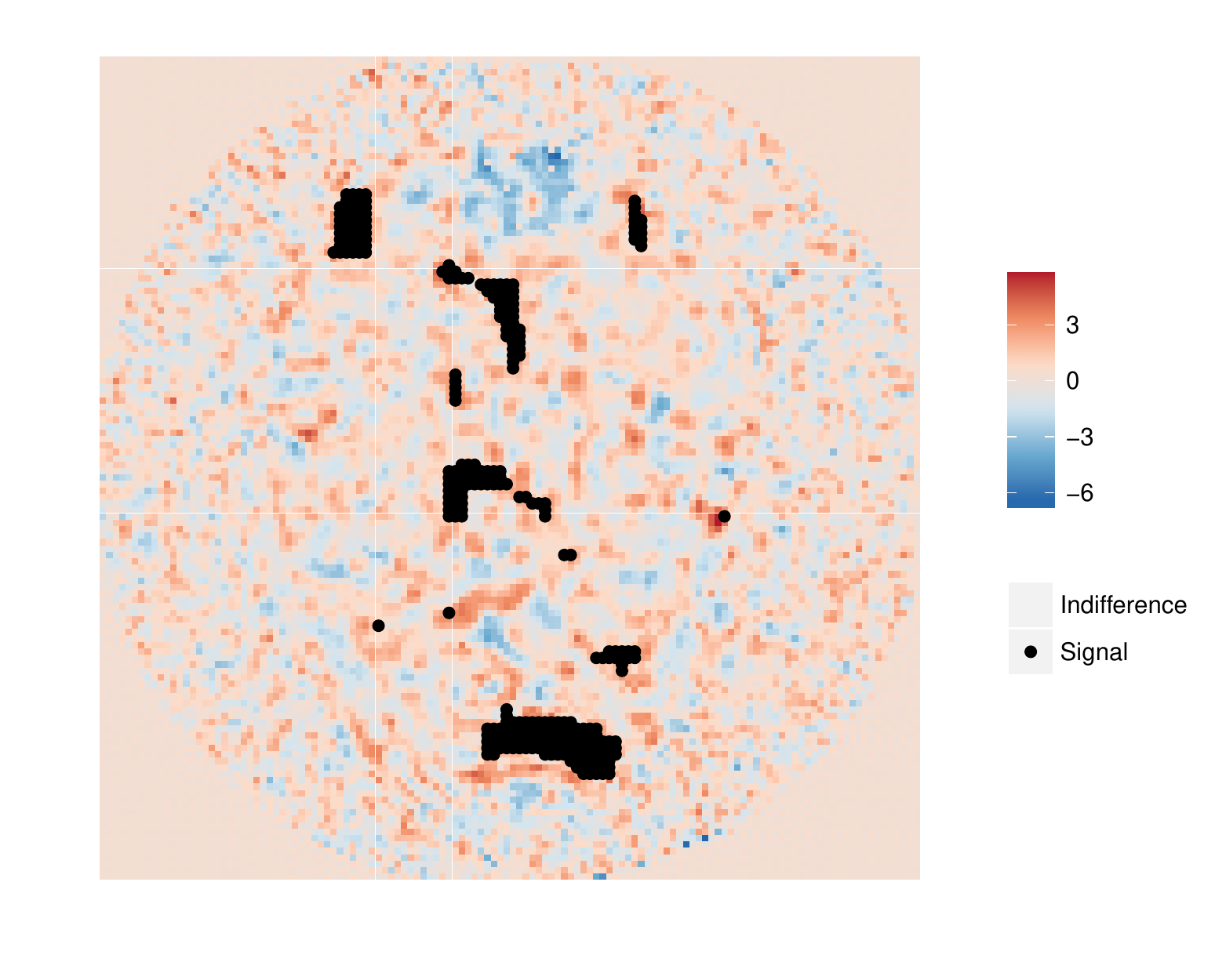}\\
\includegraphics[width=3cm]{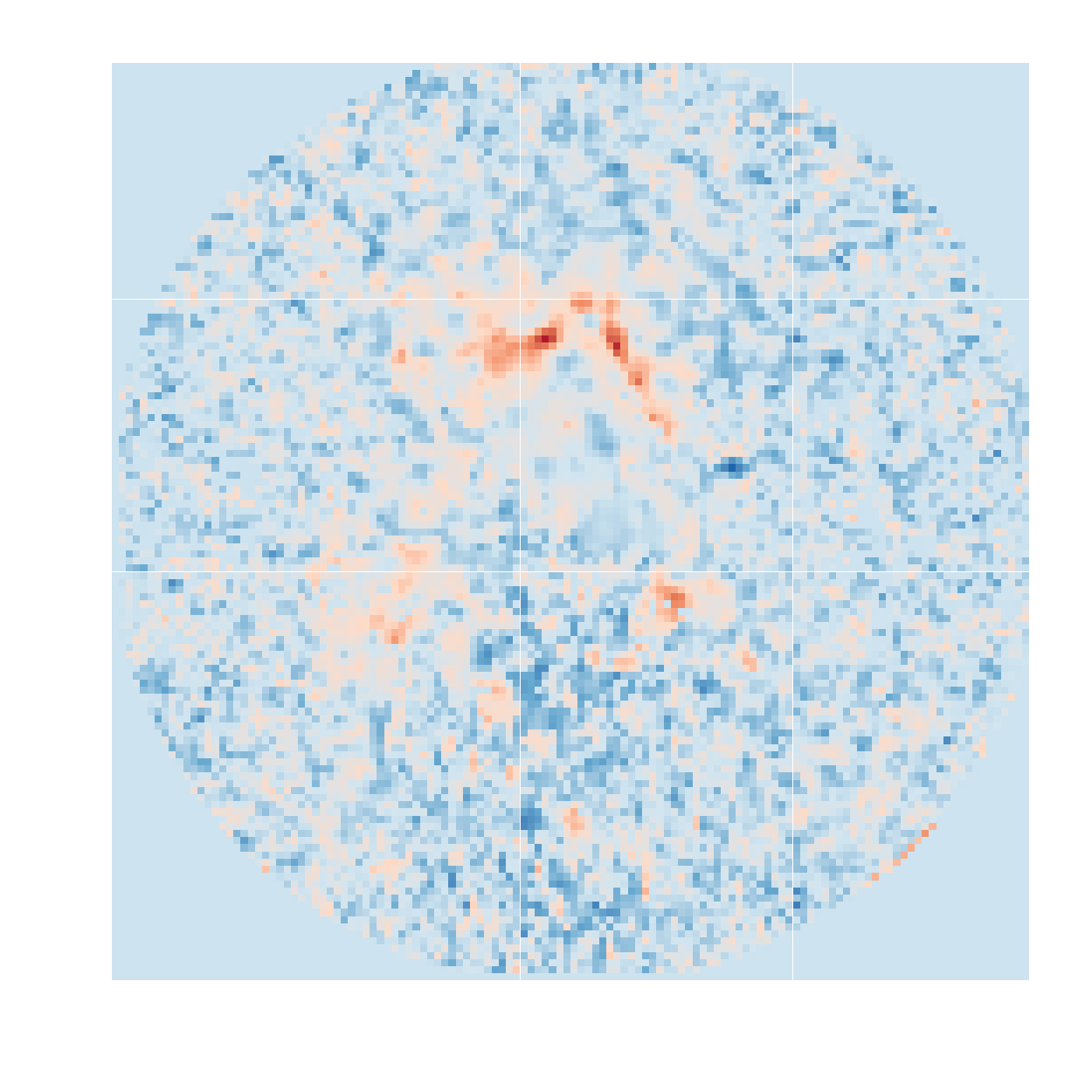}&
\includegraphics[width=3cm]{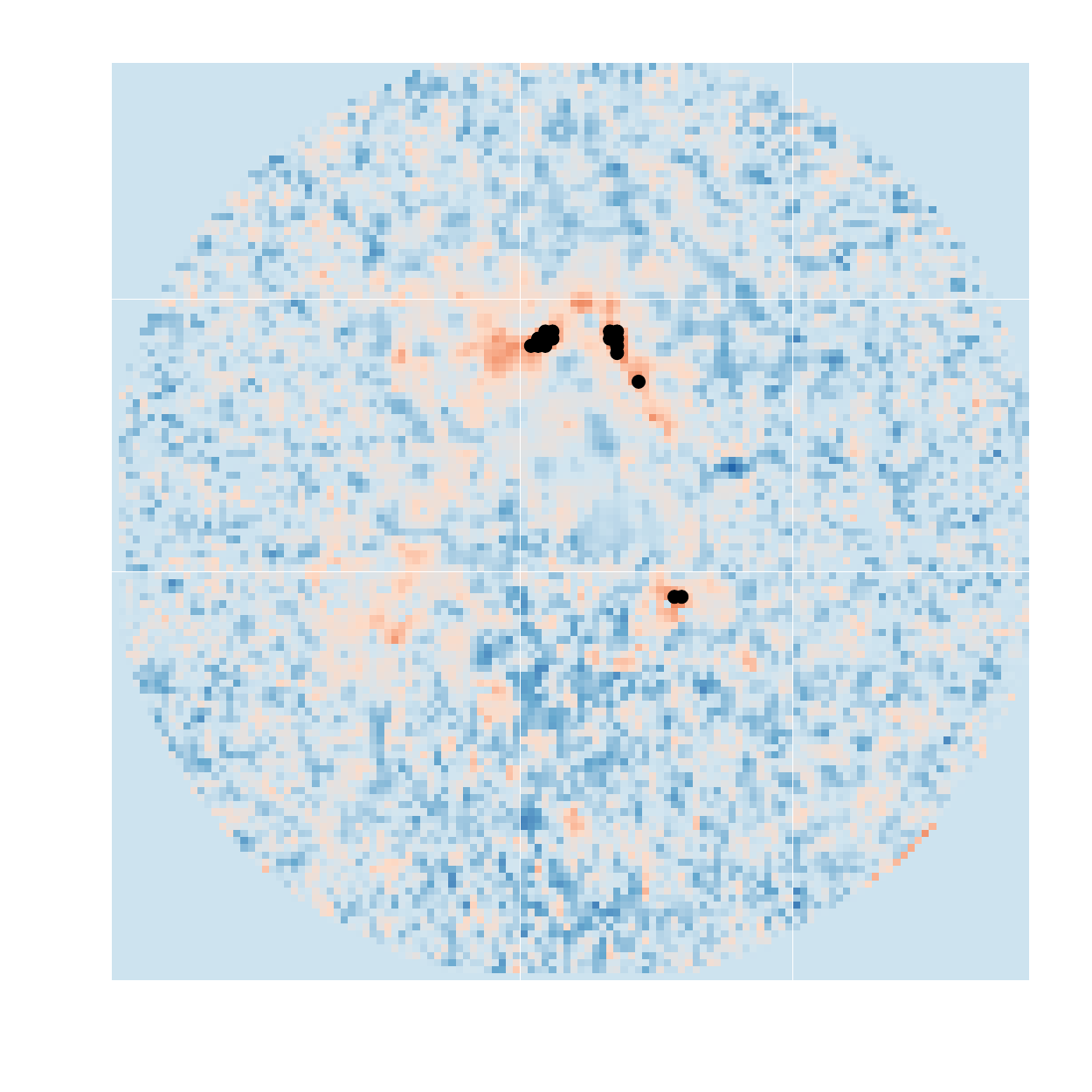}&
\includegraphics[width=3cm]{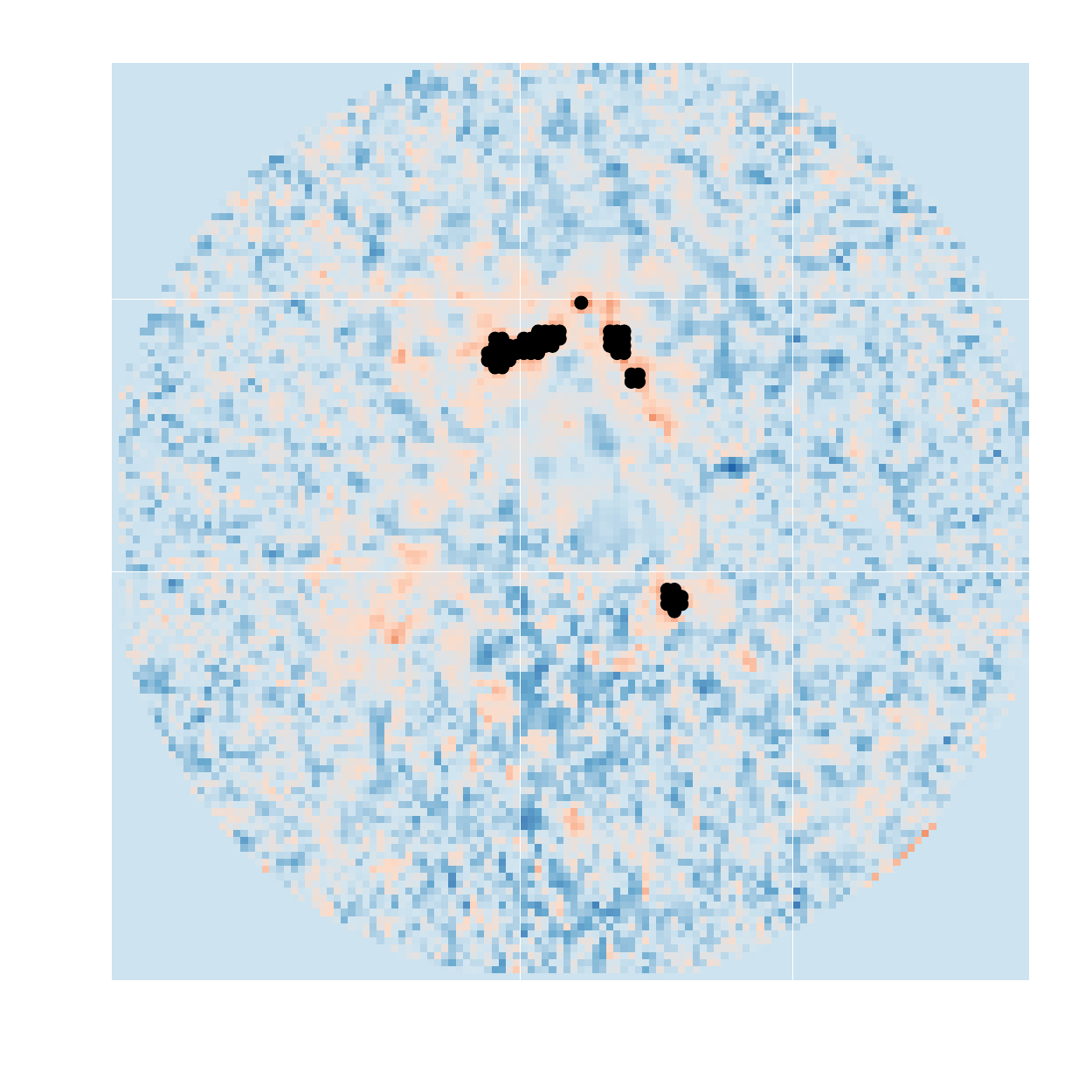}&
\includegraphics[width=3cm]{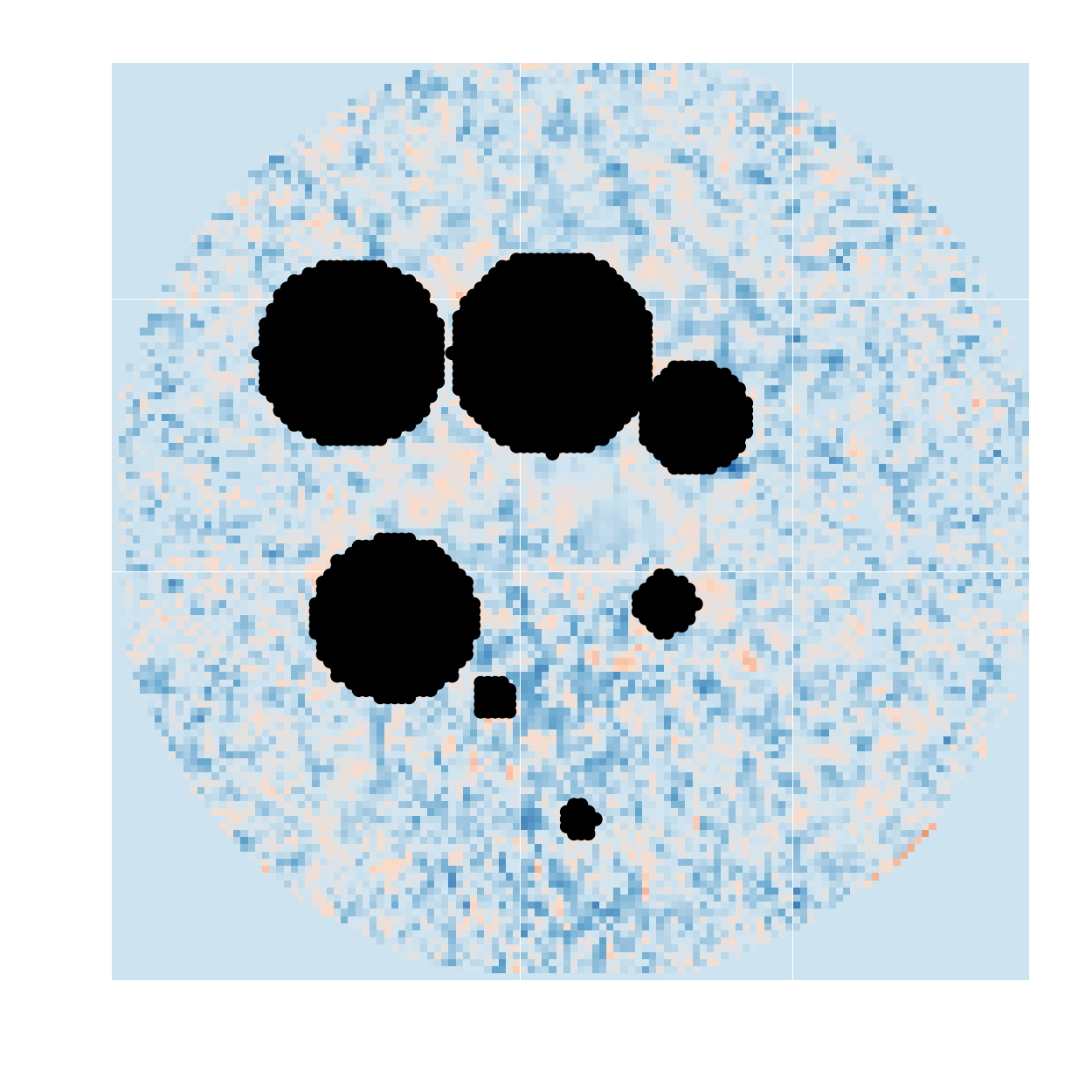}&
\includegraphics[width=3.7cm]{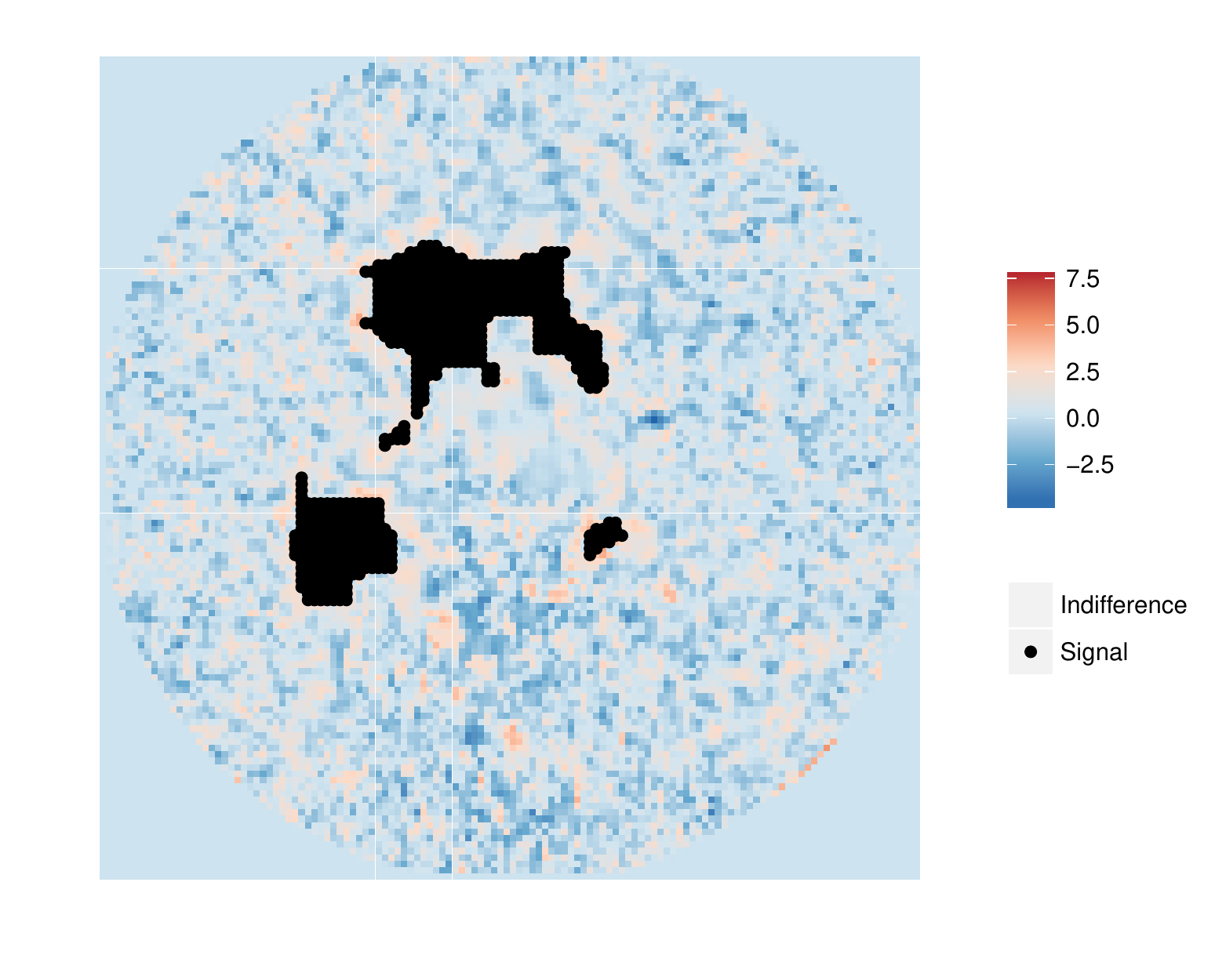}\\
\end{tabular}
\caption{Detected signal regions for the real fMRI dataset. The first column is the raw data. Conventional FDR approach (second column), FDR$_L$ approach (third column), Scan statistics (fourth column) and our proposed SCUSUM method (fifth column) are shown. The black part is the detected signal region. Here, significant level is $\alpha=0.0001.$} \label{fMRI}
\end{figure}

\section{Conclusion} \label{sec:Conclu}

In this work, we proposed a spatial signal detection method, SCUSUM, which could accommodate the local spatial information. 
SCUSUM consists of two steps: firstly signal weights are estimated by moving window projecting and CUSUM cut-off; then a threshold is determined with given significant level $\alpha$ to control marginal false discovery rate. 
Our simulation study shows that our method has a better performance compared to FDR$_L$ method.
Empirically, SCUSUM tends to detect spatially gouped and weak signals, which are missed by the other two methods.
Finally, our method is applied to a real fMRI data to illustrate its detection effectiveness.

In model \ref{basic_model}, though our method doesn't need to specify the distribution for noise process, the noise processes are assumed to be independent, and this is a strong assumption in spatial statistics.
In the future work, it could be possible to consider a dependent noise process. 
In our simulation, the result shows that with weak spaital dependent noise process, SCUSUM could still maintain it effectiveness.
Another key issue is to consider how to combine multi-source image data. 
So far, we only consider the observations are scalar and use moving window idea to project spatial observations to a sequence. 
However, this projecting method could be problematic when the observation of a spatial location is high-dimensional data or functional data.

\clearpage

\bibliographystyle{plain}
\bibliography{reference}

\clearpage

\appendix
\section{Proof for Lemma \ref{lemma.1}} \label{proof1}

\begin{proof}
In assumed model \ref{basic_model}, the noise processes are independent. 
Thus, in the $i$th block, the sampled representative $\gamma_i$ is independent with the rest observations, which implies $\gamma_i$ and $\tilde{\mu}_i$ are independent, i.e. $\gamma_i\perp\tilde{\mu}_i.$
Also the representatives and pseudo block means between blocks are indpendent. 
These lead that $\{\gamma_i\}$ and $\{\tilde{\mu}_i\}$ are indepedent.

Under the null hypothesis $H_0,$ i.e. there is no signal and $\mathbb{E}[\tilde{\mu}_i]=\mu_0,$ pseudo block mean sequence $\{\tilde{\mu}_i\}$ are i.i.d, as well as representative sequence $\{\gamma_i\}$. 
{}For the distribution of $\{\gamma_i^*\},$
we have following:
\small{
\begin{align}
f(\gamma_1^*\le y_1,...,\gamma_b^*\le y_n)&= \sum_{([1],...,[b])\in S_b} f(\gamma_{[1]}\le y_1,...,\gamma_{[b]}\le y_b|\tilde{\mu}_{[1]}\ge...\ge\tilde{\mu}_{[b]})f(\tilde{\mu}_{[1]}>...>\tilde{\mu}_{[b]})\notag\\
&=\sum_{([1],...,[b])\in S_b} f(\gamma_{[1]}\le y_1,...,\gamma_{[b]}\le y_b|\tilde{\mu}_{[1]}\ge...\ge\tilde{\mu}_{[b]})\frac{1}{b!}\notag\\
&=\sum_{([1],...,[b])\in S_b} f(\gamma_{[1]}\le y_1,...,\gamma_{[b]}\le y_b)\frac{1}{b!}\notag\\
&=\sum_{([1],...,[b])\in S_b} f(\gamma_{1}\le y_1,...,\gamma_{b}\le y_b)\frac{1}{b!} \notag\\
&=f(\gamma_{1}\le y_1,...,\gamma_{b}\le y_b) \notag.
\end{align}
}
Here $([1],...,[b])$ is the one possible decreasing order for $\{\tilde{\mu}_i\},$ $S_b$ presents the set of all the possible orders.
The first equation is according to bayesian formula;
the second equation is because of $f(\tilde{\mu}_{[1]}>...>\tilde{\mu}_{[b]})=\int_{\tilde{\mu}_{[1]}>...>\tilde{\mu}_{[b]}}dm(\tilde{\mu}_{[1]},...,\tilde{\mu}_{[b]})=1/n!$ under the independence of $\{\tilde{\mu}_i\};$
the third equation is due to independence between $\{\gamma_i\}$ and $\{\tilde{\mu}_i\};$
the fourth equation is because $\{\gamma_i\}$ are independent;
the fifth equation is due to the cardinality of $S_b$ is $1/b!.$ 
From above result, we reach that $\{\gamma_i^*\}$ are also i.i.d, having the same distribution with $\{\gamma_i\}$.

Under the alternative hypothesis $H_1,$ as the number of observations in each block $n_i$ goes to infinity, the weak law of larger number supports that $\tilde{\mu}_i \stackrel{p}{\rightarrow} \mathbb{E}[\tilde{\mu}_i].$ Thus, with (\ref{gammai}) and (\ref{tilde_mu}), (\ref{gamma*}) holds.
And $l_1$ is the number of block inside the signal region, $(b-l_2)$ is the number of block inside the indifference region and $(l_2-l_1)$ is the number of block at the boundary.
\end{proof}

\section{Proof for Theorem \ref{Theo.1}} \label{proof2}

\begin{proof}
The proof idea is similar with \cite{aue2009estimation}. W.l.o.g, here we consider the variance of noise processes is $1,$ $\mu_0=0$ and $\mu_1=\Delta.$ Define the ratio of signal region to the entire spatial domain is $\theta,$ hence, as blocks become finer and finer ($b\rightarrow \infty$), the ratio of the blocks with signal representative to the total blocks is getting closer to $\theta.$ 

Under $H_1,$ define the following events: 
\begin{align}\label{event}
\left\{
\begin{aligned}
&A_1:=\{\mu_i\le \mu_j,~ B_i \in \mathscr{D}_{\mathscr{A}}, ~ B_j \in \mathscr{D}_{\mathscr{A}^c},~ \forall i,j\},\\
&A_2:=\{\mu_k\le \mu_j,~ B_i \in \mathscr{D}_{\mathscr{A}}, ~ B_k \text{ at boundary},~ \forall i,k\},\\
&A_3:=\{\mu_j\le \mu_k,~ B_j \in \mathscr{D}_{\mathscr{A}^c}, ~ B_k \text{ at boundary},~ \forall j,k\}.\\
\end{aligned}
\right.
\end{align}
Event $(A_1\cup A_2 \cup A_3)^c$ presents the scenario that the order of $\{\gamma_i^*\}$ is from signal blocks $[1,l_1)$ to interim (boundary) blocks $[l_1,l_2]$ and then to indifferent blocks $(l_2,b]$ (see Figure \ref{fig:Ill_gamma} (b).) And with Lemma \ref{lemma.1}, $\mathbb{P}((A_1\cup A_2 \cup A_3)^c)=1,$ as $\min n_i \rightarrow \infty.$

With event $(A_1\cup A_2 \cup A_3)^c,$ following we show that the probability of cut-off location $t=\arg\max_i \tilde{\gamma}_i$ falling into $[l_1,l_2]$ would converge to $1.$ To proof that, we define statistics 
\begin{equation}
Q(r)=\tilde{\gamma}_r^2=(\sum_{i=1}^{r} \gamma^*_i-\frac{r}{b}\sum_{i=1}^{b}\gamma^*_i)^2.
\end{equation}

First consider the probability of event $B_1(N)=\{t \ge (l_2+N)\},$ with $(l_2+N)\le b,$ $N$ is a fixed constant. 
Define $R(r;l_2)=Q(r)-Q(l_2),$ and note that $Q(l_2)$ is a constant. 
\begin{align*}
R(r;l_2)&=Q(r)-Q(l_2)\\
    &=(\sum_{i=1}^{r} \gamma^*_i-\frac{r}{b}\sum_{i=1}^{b}\gamma^*_i)^2-(\sum_{i=1}^{l_2} \gamma^*_i-\frac{l_2}{b}\sum_{i=1}^{b}\gamma^*_i)^2\\
    &=\underbrace{[\sum_{l_2+1}^{r}\gamma^*_i-(r-l_2)\bar{\gamma}^*]}_{(I)}\underbrace{[\sum_{i=1}^{r}\gamma_i+\sum_{j=1}^{l_2}\gamma_j-(r+l_2)\bar{\gamma}^*]}_{(II)}
\end{align*}
where $\bar{\gamma}^*=\frac{1}{b}\sum_{i=1}^{b}\gamma^*_i.$
And with equation (\ref{gamma*}), following equations hold:
\begin{align}
(I)&=\sum_{i=l_2+1}^{r}\epsilon_i-(r-l_2)\frac{1}{b}\sum_{i=1}^{b}\epsilon_i-(r-l_2)\theta\Delta,\label{*}\\
(II)&=\sum_{i=1}^{r}\epsilon_i+\theta b\Delta+\sum_{j=1}^{l_2}\epsilon_j+\theta b\Delta-(r+l_2)\frac{1}{b}\sum_{i=1}^{b}\epsilon_i-(r+l_2)\theta\Delta \notag\\
&=\sum_{i=1}^{r}\epsilon_i+\sum_{j=1}^{l_2}\epsilon_j-(r+l_2)\frac{1}{b}\sum_{i=1}^{b}\epsilon_i+(2\theta b-(r+l_2)\theta)\Delta,
\end{align}
Define the following statistics:
\begin{align}
\left\{
\begin{aligned}
&E^1(r;l_2):=\sum_{i=l_2+1}^{r}\epsilon_i-(r-l_2)\frac{1}{b}\sum_{i=1}^{b}\epsilon_i,\\
&E^2(r;l_2):=\sum_{i=1}^{r}\epsilon_i+\sum_{j=1}^{l_2}\epsilon_j-(r+l_2)\frac{1}{b}\sum_{i=1}^{b}\epsilon_i,\\
&D^1(r;l_2):=-(r-l_2)\theta\Delta,\\
&D^2(r;l_2):=(2\theta b-(r+l_2)\theta)\Delta.\\
\end{aligned}
\right.
\end{align}
So $(I)=E^1(r;l_2)+D^1(r;l_2)$ and $(II)=E^2(r;l_2)+D^2(r;l_2).$ 

As $b\rightarrow \infty,$  we have 
\begin{align*}
\max_{(l_2+N)\le r \le b} D^1(r;l_2)D^2(r;l_2) &= \max_{(l_2+N)\le r \le b} [-(r-l_2)\theta\Delta][(2\theta b-(r+l_2)\theta)\Delta]\\
&= \max_{(l_2+N)\le r \le b} -\theta^2\Delta^2(r-l_2)(2-\frac{r+l_2}{b})b\\
&= -\theta^2\Delta^2N(2-\frac{2l_2+N}{b})b,
\end{align*}
the last equation is due to the $D^1(r;l_2)D^2(r;l_2)$ reaches the maximum with $r=(l_2+N).$ Also $\forall \epsilon\ge 0,$ we have
\begin{align*}
&\lim_{b\rightarrow \infty}\sup\mathbb{P}(\max_{(l_2+N)\le r \le b} D^1(r;l_2)D^2(r;l_2) >-\epsilon)\\
&=\lim_{b\rightarrow \infty}\sup \mathbb{P}(-\theta^2\Delta^2N(2-\frac{2l_2+N}{b})b>-\epsilon)\\
&=\lim_{b\rightarrow \infty} \sup \mathbb{P}(\theta^2\Delta^2N(2-\frac{2l_2+N}{b})b\le \epsilon)=0\\
\end{align*}
If we could prove $D^1D^2(r;l_2)$ is the leading term in $R(r;l_2),$ then
\begin{align*}
\lim_{b\rightarrow \infty} \sup\mathbb{P}(B_1(N)) &= \lim_{b\rightarrow \infty} \sup\mathbb{P}(t\ge l_2+N)\\
&=\lim_{b\rightarrow \infty} \sup\mathbb{P}(\max_{(l_2+N)\le r \le b}R(r;l_2)>0)\\
&=\lim_{b\rightarrow \infty} \sup \mathbb{P}(\theta^2\Delta^2N(2-\frac{2l_2+N}{b})b +O(1)\le \epsilon)=0
\end{align*}
Hence $\mathbb{P}(t\ge l_2)=\cup_{N=0}^{b-l_2}\mathbb{P}(B_1(N))=0.$

Following lemmas support that $D^1(r;l_2)D^2(r;l_2)$ is the leading term in $R(r;l_2),$ with $r\in[l_2+N,b].$

\begin{lemma}\label{part1}
With the assumptions of Theorem \ref{Theo.1}, given $N,$ $\forall \epsilon>0,$
\begin{equation}
\lim_{b\rightarrow \infty} \sup  \mathbb{P}(\max_{(l_2+N)\le r \le b} \frac{|E^1(r;l_2)E^2(r;l_2)|}{|D^1(r;l_2)D^2(r;l_2)|} \ge \epsilon)=0.
\end{equation}
\end{lemma}
\begin{proof} With the brief derivation, we have,
\begin{align*}
&\max_{(l_2+N)\le r \le b} \frac{|E^1(r;l_2)E^2(r;l_2)|}{|D^1(r;l_2)D^2(r;l_2)|} \\
&\le \max_{(l_2+N)\le r \le b} \frac{|E^1(r;l_2)||E^2(r;l_2)|}{\theta\Delta^2(r-l_2)(2b\theta-(l_2+r))}\\
&=O(1)\max_{(l_2+N)\le r \le b}\frac{|E^1(r;l_2)||E^2(r;l_2)|}{(r-l_2)(l_2+r)}\\
&=O(1)\max_{(l_2+N)\le r \le b}\underbrace{\frac{|E^1(r;l_2)|}{(r-l_2)}}_{(III)}\underbrace{\frac{|E^2(r;l_2)|}{(l_2+r)}}_{(IV)}
\end{align*}
the last equation is because of $l_2>\theta\Delta.$

For $(III),$
\begin{align*}
\frac{|E^1(r;l_2)|}{(r-l_2)}&\le \frac{|\sum_{i=l_2+1}^{r}\epsilon_i|+|(r-l_2)\frac{1}{b}\sum_{i=1}^{b}\epsilon_i|}{(r-l_2)}\\
&=\frac{|\sum_{i=l_2+1}^{r}\epsilon_i|}{(r-l_2)}+|\frac{1}{b}\sum_{i=1}^{b}\epsilon_i| \stackrel{p}{\rightarrow} 0,
\end{align*}
the last equation is due to the law of iterated logarithm and the weak law of large number.
Similarly, for $(IV),$
\begin{align*}
\frac{|E^2(r;l_2)|}{(l_2+r)}&\le \frac{|\sum_{i=1}^{r}\epsilon_i+\sum_{j=1}^{l_2}\epsilon_j|+|(r+l_2)\frac{1}{b}\sum_{i=1}^{b}\epsilon_i|}{(l_2+r)}\stackrel{p}{\rightarrow} 0.
\end{align*}
Hence, with continuous mapping theorem, we have the result.
\end{proof}

\begin{lemma}\label{part2}
With the assumptions of Theorem \ref{Theo.1}, given $N,$ $\forall \epsilon>0,$
\begin{equation}
\lim_{b\rightarrow \infty} \sup  \mathbb{P}(\max_{(l_2+N)\le r \le b} \frac{|E^1(r;l_2)D^2(r;l_2)|}{|D^1(r;l_2)D^2(r;l_2)|} \ge \epsilon)=0.
\end{equation}
\end{lemma}
\begin{proof}
Similarly, we have
\begin{align*}
&\max_{(l_2+N)\le r \le b} \frac{|E^1(r;l_2)D^2(r;l_2)|}{|D^1(r;l_2)D^2(r;l_2)|} \\
&\le O(1) \max_{(l_2+N)\le r \le b} \frac{|E^1(r;l_2)|}{(r-l_2)}\frac{|D^2(r;l_2)|}{(r+l_2)}.
\end{align*}
From Lemma \ref{part1}, we know that $\frac{|E^1(r;l_2)|}{(r-l_2)}\stackrel{p}{\rightarrow} 0.$
With directly derivation, $\frac{|D^2(r;l_2)|}{(r+l_2)}=O(1).$
Hence, $\max_{(l_2+N)\le r \le b} \frac{|E^1(r;l_2)D^2(r;l_2)|}{|D^1(r;l_2)D^2(r;l_2)|}\stackrel{p}{\rightarrow} 0.$
\end{proof}

\begin{lemma}\label{part3}
With the assumptions of Theorem \ref{Theo.1}, given $N,$ $\forall \epsilon>0,$
\begin{equation}
\lim_{b\rightarrow \infty} \sup  \mathbb{P}(\max_{(l_2+N)\le r \le b} \frac{|D^1(r;l_2)E^2(r;l_2)|}{|D^1(r;l_2)D^2(r;l_2)|} \ge \epsilon)=0.
\end{equation}
\end{lemma}
\begin{proof}
The idea is the same with Lemma \ref{part2}:
\begin{equation}
\max_{(l_2+N)\le r \le b} \frac{|D^1(r;l_2)E^2(r;l_2)|}{|D^1(r;l_2)D^2(r;l_2)|} \le O(1) \max_{(l_2+N)\le r \le b} \frac{|D^1(r;l_2)|}{(r-l_2)}\frac{|E^2(r;l_2)|}{(r+l_2)} \stackrel{p}{\rightarrow} 0,
\end{equation}
with $\frac{|D^1(r;l_2)|}{(r-l_2)}=O(1)$ and $\frac{|E^2(r;l_2)|}{(r+l_2)}\stackrel{p}{\rightarrow} 0.$
\end{proof}

The above three lemmas support that $D^1(r;l_2)D^2(r;l_2)$ is the leading term in $R(r;l_2),$ with $r\in[l_2+N,b].$ 
Hence, the probability of event $\{t>l_2\}\cup (A_1\cup A_2 \cup A_3)\rightarrow 0,$ as $b\rightarrow \infty$ and $\min n_i\rightarrow \infty.$

For the other side, consider event $B_2(N)=\{t\le(l_1-N)\},$ with a given $N$ and $l_1\ge N.$ 
Similarly define $R(r;l_1)=Q(r)-Q(l_1),$ and following we will show $R(r;l_1)<0,~\forall r<l_1$ asymptotically with probability $1.$ 
\begin{align*}
R(r;l_1)&=Q(r)-Q(l_1)\\
    &=(\sum_{i=1}^{r} \gamma^*_i-\frac{r}{b}\sum_{i=1}^{b}\gamma^*_i)^2-(\sum_{i=1}^{l_1} \gamma^*_i-\frac{l_1}{b}\sum_{i=1}^{b}\gamma^*_i)^2\\
    &=\underbrace{[-\sum_{r+1}^{l_1}\gamma^*_i-(r-l_1)\bar{\gamma}^*]}_{(V)}\underbrace{[\sum_{i=1}^{r}\gamma_i+\sum_{j=1}^{l_1}\gamma_j-(r+l_1)\bar{\gamma}^*]}_{(VI)}
\end{align*}

Define the following statistics:
\begin{align}
\left\{
\begin{aligned}
&E^1(r;l_1):=-\sum_{i=r+1}^{l_1}\epsilon_i-(r-l_1)\bar{\gamma}^*,\\
&E^2(r;l_1):=\sum_{i=1}^{r}\epsilon_i+\sum_{j=1}^{l_1}\epsilon_j-(r+l_1)\bar{\gamma}^*,\\
&D^1(r;l_1):=-(l_1-r)(1-\theta)\Delta,\\
&D^2(r;l_1):=(r+l_1)(1-\theta)\Delta.\\
\end{aligned}
\right.
\end{align}
Similarly, we have $(V)=E^1(r;l_1)+D^1(r;l_1),$ $(VI)=E^1(r;l_2)+D^1(r;l_2).$
Also, as $b\rightarrow \infty,$ $l_1=\theta b\rightarrow \infty,$ so we have 
\begin{align*}
\max_{1\le r \le (l_1-N)} D^1(r;l_1)D^2(r;l_1) &= [-N(r+l_1)(1-\theta)^2\Delta^2]
\end{align*}
and 
\begin{align*}
&\lim_{b\rightarrow \infty}\sup\mathbb{P}(\max_{1\le r \le (l_1-N)} D^1(r;l_1)D^2(r;l_1) >-\epsilon)\\
&=\lim_{b\rightarrow \infty}\sup \mathbb{P}(-N(r+l_1)(1-\theta)^2\Delta^2>-\epsilon)\\
&=\lim_{b\rightarrow \infty} \sup \mathbb{P}(N(r+\theta b)(1-\theta)^2\Delta^2 \le \epsilon)=0\\
\end{align*}

Following lemma shows that $D^1(r;l_1)D^2(r;l_1)$ is the leading term in $R(r;l_1),$ with $r\in[1,l_1-N].$

\begin{lemma}\label{part4}
With the assumptions of Theorem \ref{Theo.1}, given $N,$ $\forall \epsilon>0,$ we have 
\begin{align}
\left\{
\begin{aligned}
&\lim_{b\rightarrow \infty} \sup  \mathbb{P}(\max_{1\le r \le (l_1-N)} \frac{|E^1(r;l_1)E^2(r;l_1)|}{|D^1(r;l_1)D^2(r;l_1)|} \ge \epsilon)=0,\\
&\lim_{b\rightarrow \infty} \sup  \mathbb{P}(\max_{1\le r \le (l_1-N)} \frac{|D^1(r;l_1)E^2(r;l_1)|}{|D^1(r;l_1)D^2(r;l_1)|} \ge \epsilon)=0,\\
&\lim_{b\rightarrow \infty} \sup  \mathbb{P}(\max_{1\le r \le (l_1-N)} \frac{|E^1(r;l_1)D^2(r;l_1)|}{|D^1(r;l_1)D^2(r;l_1)|} \ge \epsilon)=0.
\end{aligned}
\right.
\end{align}
\end{lemma}
\begin{proof}
Similarly with Lemma \ref{part1}-\ref{part3}, we have following:
\begin{equation}
\left\{
\begin{aligned}
&\frac{|E^1(r;l_1)|}{(l_1-r)}\le \frac{|\sum_{i=r+1}^{l_1}\gamma_i^*|}{l_1-r}+|\frac{1}{b}\sum_{i=1}^{b}\gamma_i^*|\stackrel{p}{\rightarrow} 0;\\
&\frac{|E^2(r;l_1)|}{(l_1+r)}\le \frac{|\sum_{i=1}^{r}\gamma_i^*+\sum_{i=1}^{l_1}\gamma_i^*|}{l_1+r}+|\frac{1}{b}\sum_{i=1}^{b}\gamma_i^*|\stackrel{p}{\rightarrow} 0;\\
&\frac{|D^1(r;l_1)|}{(l_1-r)} =\frac{|D^2(r;l_1)|}{(l_1+r)} =  O(1).
\end{aligned}
\right.
\end{equation}

Hence, we have the results
\begin{equation}
\left\{
\begin{aligned}
&\max_{1\le r \le (l_1-N)} \frac{|E^1(r;l_1)E^2(r;l_1)|}{|D^1(r;l_1)D^2(r;l_1)|} 
\le O(1)\max_{1\le r \le (l_1-N)} \frac{|E^1(r;l_1)||E^2(r;l_1)|}{(l_1-r)(l_1+r)}\stackrel{p}{\rightarrow} 0;\\
&\max_{1\le r \le (l_1-N)} \frac{|D^1(r;l_1)E^2(r;l_1)|}{|D^1(r;l_1)D^2(r;l_1)|}
\le O(1)\max_{1\le r \le (l_1-N)} \frac{|D^1(r;l_1)||E^2(r;l_1)|}{(l_1-r)(l_1+r)}\stackrel{p}{\rightarrow} 0;\\
&\max_{1\le r \le (l_1-N)} \frac{|E^1(r;l_1)D^2(r;l_1)|}{|D^1(r;l_1)D^2(r;l_1)|}
\le O(1)\max_{1\le r \le (l_1-N)} \frac{|E^1(r;l_1)||D^2(r;l_1)|}{(l_1-r)(l_1+r)}\stackrel{p}{\rightarrow} 0.
\end{aligned}
\right.
\end{equation}

\end{proof}
With above conclusion, we have 
\begin{align*}
\lim_{b\rightarrow \infty} \sup\mathbb{P}(B_2(N)) &= \lim_{b\rightarrow \infty} \sup\mathbb{P}(t\le l_1-N)\\
&=\lim_{b\rightarrow \infty} \sup\mathbb{P}(\max_{1\le r \le (l_1-N)}R(r;l_1)>0)\\
&=\lim_{b\rightarrow \infty} \sup \mathbb{P}(N(r+\theta b)(1-\theta)^2\Delta^2 +O(1)\le \epsilon)=0
\end{align*}
which implies $\mathbb{P}(t\le l_1)=\cup_{N=1}^{l_1}\mathbb{P}(B_2(N))=0.$

With these results, we have 
\begin{align*}
&\lim_{b\rightarrow \infty} \mathbb{P}(\{t \in [l_1,l_2]\})\\
&\le 1- \lim_{b\rightarrow \infty} \sup \mathbb{P}(\{t<l_1\} \cup \{t>l_2\}\cup (A_1\cup A_2 \cup A_3))\\
&\rightarrow 1,
\end{align*} 
as $b\rightarrow \infty$ and $\min n_i\rightarrow \infty.$

\end{proof}

\section{Proof for Lemma \ref{lemma.sym}} \label{proof3}

\begin{proof}
This lemma is easy to prove: from Lemma \ref{lemma.1}, as $\min n_i \rightarrow \infty,$ $\{\gamma_i^*\}$ are i.i.d. Consider sequence $\{y_1,...,y_b\},$ then $\mathbb{P}(\gamma_1^*=y_1,...,\gamma_b^*=y_b)=\mathbb{P}(\gamma_1^*=y_b,...,\gamma_b^*=y_1).$ Note that in the left side $\{\gamma_i^*\}$ equal to the reversed sequence $\{y_b,...,y_1\}.$ Via applying CUSUM cut-off on $\{y_1,...,y_b\}$ and $\{y_b,...,y_1\},$ the detection result is opposite. Hence the distribution for signal weights under $H_0$ is symmetric under the asymptotical setting.

\end{proof}

\end{document}